\documentclass[a4paper, 11pt, oneside,reqno]{article}

\usepackage[utf8]{inputenc}
\usepackage[english]{babel}
\usepackage{caption}
\usepackage{float}
\usepackage{graphicx}
\usepackage{array}
\usepackage{hyperref} 

\usepackage{makecell} 

\usepackage[left=1in, right=1in, top=1in, bottom=1in, includefoot, headheight=13.6pt]{geometry}
\usepackage{url}
\usepackage{datetime}
\usepackage{times}
\usepackage[dvipsnames]{xcolor}
\usepackage{amsmath}
\usepackage{amssymb}
\usepackage{booktabs} 
\usepackage[titletoc,title]{appendix} 

\usepackage{tikz}

\usepackage{geometry} 

\usepackage[amsmath,thmmarks]{ntheorem} 

\newtheorem{theorem}{Theorem}

\newtheorem{remark}[theorem]{Remark}

\newtheorem{lemma}[theorem]{Lemma}
\newtheorem{definition}[theorem]{Definition}

\newtheorem{proposition}[theorem]{Proposition}

\theoremstyle{nonumberplain}
\theorembodyfont{\normalfont}
\theoremsymbol{\ensuremath{\square}}
\newtheorem{proof}{Proof}

\makeatletter
\newcommand{\leqnomode}{\tagsleft@true}
\newcommand{\reqnomode}{\tagsleft@false}
\makeatother

\makeatletter
\newcommand*\rel@kern[1]{\kern#1\dimexpr\macc@kerna}
\newcommand*\widebar[1]{%
  \begingroup
  \def\mathaccent##1##2{%
    \rel@kern{0.8}%
    \overline{\rel@kern{-0.8}\macc@nucleus\rel@kern{0.2}}%
    \rel@kern{-0.2}%
  }%
  \macc@depth\@ne
  \let\math@bgroup\@empty \let\math@egroup\macc@set@skewchar
  \mathsurround\z@ \frozen@everymath{\mathgroup\macc@group\relax}%
  \macc@set@skewchar\relax
  \let\mathaccentV\macc@nested@a
  \macc@nested@a\relax111{#1}%
  \endgroup
}
\makeatother

\newcommand\PBS[1]{\let\temp=\\%
  #1%
  \let\\=\temp
}

\newcommand{\subf}[2]{%
  {\small\begin{tabular}[t]{@{}c@{}}
   #1\\ #2
  \end{tabular}}%
}

\newdateformat{dateenglish}{\monthname[\THEMONTH]~\THEDAY,~\THEYEAR}
\DeclareMathOperator*{\argmin}{argmin}
\DeclareMathOperator{\sgn}{sgn}
\DeclareMathOperator{\Var}{Var} 
\DeclareMathOperator{\Cov}{Cov}
\DeclareMathOperator{\Cor}{Cor}
\newcommand{\VaR}{\operatorname{VaR}}
\newcommand{\ES}{\operatorname{ES}}

\DeclareMathOperator{\E}{\mathbb{E}}
\def\N{\mathbb{N}}
\def\Z{\mathbb{Z}}
\def\R{\mathbb{R}}
\newcommand{\1}{\mbox{1\hspace{-0.28em}I}}
\newcommand\numberthis{\addtocounter{equation}{1}\tag{\theequation}}

\title{The Impact of the Choice of Risk and Dispersion Measure on Procyclicality}
\author{\Large Marcel Br\"autigam$\dagger$ $\ddagger^{\;\ast}$ and Marie Kratz$\dagger^{\;\ast}$\\[1ex]
\small $\dagger$ ESSEC Business School Paris, CREAR \\ \small $\ddagger$ Sorbone University, LPSM \\ \small $^{\;\ast}$ LabEx MME-DII}

\date{}

\begin{document}

\maketitle

\begin{abstract}
\noindent
Procyclicality of historical risk measure estimation means that one tends to over-estimate future risk when present realized volatility is high and vice versa under-estimate future risk when the realized volatility is low.
Out of it different questions 
arise, relevant for applications and theory: 
What are the factors which affect the degree of procyclicality? More specifically, how does the choice of risk measure affect this? How does this behaviour vary with the choice of realized volatility estimator?
How do different underlying model assumptions influence the pro-cyclical effect?

In this paper we consider three different well-known risk measures (Value-at-Risk, Expected Shortfall, Expectile), the r-th absolute centred sample moment, for any integer $r>0$, as realized volatility estimator (this includes the sample variance and the sample mean absolute deviation around the sample mean) and two models (either an iid model or an augmented GARCH($p$,$q$) model).
We show that the strength of procyclicality depends on these three factors, the choice of risk measure, the realized volatility estimator and the model considered.
But, no matter the choices, the procyclicality will always be present.
\bigskip\newline%
\noindent  {\emph 2010 AMS classification}: 60F05; 
62H20; 62P05; 62P20; 91B30

\end{abstract}
\newpage
\tableofcontents

\newpage

\section{Introduction and Notation}
\label{sec:Intro}

The introduction of risk based solvency regulations has brought the need for financial institutions to evaluate their risk on the basis of probabilistic models. 
Two of the most popular risk measures to determine the capital needed by companies to cover their risk are Value-at-Risk (which goes back to~\cite{JPMorgan96}) and the Expected Shortfall (see~\cite{Acerbi02} and \cite{Rockafellar02}). 
The question of the appropriateness of the risk measure to use for evaluating the risk of financial institutions has been heavily debated especially after the financial crisis of 2008/2009. 
For a review of the arguments on this subject, we refer e.g. to \cite{Chen18} and \cite{Emmer15}. 
%

Independently from the choice of an adequate risk measure, there is an accepted idea that risk measurements are pro-cyclical: in times of crisis, they overestimate the future risk, while they underestimate it in quiet times. 
For a general review of the topic of pro-cyclicality, mostly in a macro-economical context, we refer to \cite{Athanasoglou14} or \cite{Brautigam19_emp} and the references therein. 
In this paper however we do not want to take any macro-economical point of view, but analyse further the pro-cyclicality of risk measures.  
Such pro-cyclicality is usually assumed to be a consequence of the volatility clustering and its return to the mean.  

In an empirical study 
%
on 11 stock indices (SI) of major economies \cite{Brautigam19_emp} conclude, 
that the pro-cyclicality can be explained by two factors: (i) the way risk is measured as a function of quantiles estimated on past observations, and (ii) the clustering and return-to-the-mean of volatility. By this on the one hand confirming the assumptions and on the other hand showing that there is an intrinsic component to the historical risk measure estimation.
Complementary work from a theoretical side are \cite{Brautigam19_iid} and \cite{Brautigam19_garch} which prove bivariate asymptotic distributions between the sample quantile and the r-th measure of dispersion in the case of an underlying iid sample or an augmented GARCH($p$,$q$) process respectively. This way the empirical observations in \cite{Brautigam19_emp} can be related to a theoretical foundation.

From these empirical and theoretical findings, a few questions naturally arise: 
Is this pro-cyclicality linked to VaR or does it more generally also apply to other well-known risk measures like ES and expectiles?
How is it influenced by the choice of measure of dispersion and how does it behave under different underlying models one would sample from?
And in general, what consequences does this imply when working with historical estimation of risk measures in practice?

In this paper we show, theoretically and empirically, that the strength of procyclicality depends on the choice of risk measure, the measure of dispersion and the model considered.
But, no matter the choices, the procyclicality will always be present.

Let us end with a remark about the structure of the paper. We finish this introduction with the necessary notation and mathematical framework (formalizing the notion of pro-cyclicality in an asymptotic setting), as well as recalling the notions of the three risk measures under consideration, VaR, ES and expectiles, and their corresponding estimators.
In Section~\ref{sec:procycl-iid} we prove the pro-cyclicality of the different risk and dispersion measures for an underlying iid model.
%
Note that assessing the pro-cyclicality in the iid case is intuitively clear: As we will recall, the risk measure estimator at time $t+1$ year is computed on a sample disjoint from the sample for the risk measure estimator at time $t$. Hence, in an iid sample, those estimators should be uncorrelated.

The pro-cyclicality in the case of augmented GARCH($p$, $q$) processes is treated in Section~\ref{sec:procycl-garch}.
As we do not have an underlying independent sample, two estimators computed on disjoint samples can still be dependent. But we show that, as in the iid case, asymptotically the risk measure estimator at time $t+1$ year will be uncorrelated to the risk measure estimator (and the measure of dispersion estimator) at time $t$.

The theoretical results are applied in Section~\ref{sec:exp}. We compare the pro-cyclicality of the different risk measures when considering the sample variance and sample MAD as measures of dispersion. 
Since only in the iid case (and not for augmented GARCH($p$, $q$) processes) closed-form analytical expressions are available, we focus on the former case, considering the Gaussian and the Student-t distribution as examples.
As a second application, we discuss the relevance of these asymptotic results in view of the empirical results on real data obtained in \cite{Brautigam19_emp}: 
Looking at the residuals of the fitted GARCH($1, 1$) process on each of the 11 indices considered, we compare their pro-cyclicality to the one of iid realizations. We show that they are similar (in the sense that the empirical pro-cyclicality of the residuals often falls within the confidence interval of the IID pro-cyclicality).
Thus, we provide additional arguments 
why we can relate the pro-cyclicality observed empirically partly to an intrinsic part as in the iid models and partly to the GARCH effects as claimed in \cite{Brautigam19_emp}.
We conclude in Section~\ref{sec:conclusio}.


{\sf \bf \large Notation}

Let $(X_1,\cdots,X_n)$ be a sample of size $n$. Assuming the random variables $X_i$'s have a common distribution, denote their parent random variable (rv) $X$ with parent cumulative distribution function (cdf) $F_X$, (and, given they exist,) probability density function (pdf) $f_X$, mean $\mu$, variance $\sigma^2$, as well as, for any integer $r\geq 1$ the r-th absolute centred moment, $\mu(X,r) := \E[\lvert X - \mu \rvert^r$ and quantile of order $p$ defined as $q_X(p):= \inf \{ x \in \R: F_X(x) \geq p \}$.
We denote the ordered sample by $X_{(1)}\leq ...\leq X_{(n)}$.

We consider the sample estimators of the two quantities of interest, i.e. first the sample quantile for any order $p \in [0,1]$ defined as $ q_n (p) = X_{( \lceil np \rceil )}$, where $\lceil x \rceil =   \min{ \{ m \in \mathbb{Z}  : m \geq x \} }$, $\lfloor x \rfloor =   \max{ \{ m \in \mathbb{Z}  : m \leq x \} }$ and $[x]$, are the rounded-up, rounded-off integer-parts and the nearest-integer of a real number $x \in \R$, respectively. Second, the r-th absolute centred sample moment defined, for $r \in \N$, by
\begin{equation}\label{eq:sampleMDisp}
\hat{m}(X,n,r) := \frac{1}{n} \sum_{i=1}^n \lvert X_i -  \bar{X}_n \rvert^r,
\end{equation}
$\bar{X}_n$ denoting the empirical mean.
Special cases of this latter estimator include the sample variance ($r=2$) 
and the sample mean absolute deviation around the sample mean ($r=1$).

Recall the standard notation $u^T$ for the transpose of a vector $u$ and,
for the signum function, $\displaystyle \sgn(x) := -\1_{(x<0)}+\1_{(x>0)}$.
Moreover the notations $\overset{d}\rightarrow$, $\overset{a.s.}\rightarrow$, $\overset{P}\rightarrow$ and $\overset{D_d[0,1]}\rightarrow$ correspond to the convergence in distribution, almost surely, in probability and in distribution of a random vector in the d-dimensional Skorohod space $D_d[0,1]$. Further, for real-valued functions $f,g$, we write $f(x) = O(g(x))$ (as $x \rightarrow \infty)$ if and only if there exists a positive constant $M$ and a real number $x_0$ s.t. $\lvert f(x) \rvert \leq M g(x)$ for all $x \geq x_0$, and $f(x)=o(g(x))$ (as $x \rightarrow \infty$) if for all $\epsilon>0$ there exists a real number $x_0$ s.t. $\lvert f(x) \rvert \leq \epsilon g(x)$ for all $x \geq x_0$. Analogously, for a sequence of rv's $X_n$ and constants $a_n$, we denote by $X_n = o_P(a_n)$ the convergence in probability to 0 of $X_n/a_n$.

{\sf \bf \large Family of Processes Considered}

As mentioned, the samples considered will be either realizations from an underlying iid distribution or from augmented GARCH($p$, $q$) processes (with the latter naturally including the former as a special case).
Such an augmented GARCH($p$, $q$) process $X=(X_t)_{t \in \Z}$, due to Duan in~\cite{Duan97}, satisfies, for integers $p \geq 1$ and $q\geq 0$,

\begin{align}
X_t &= \sigma_t ~\epsilon_t \label{eq:augm_GARCH_pq_1}
\\ \text{with~} \quad \Lambda(\sigma_t^2) &= \sum_{i=1}^p g_i (\epsilon_{t-i}) + \sum_{j=1}^q c_j (\epsilon_{t-j}) \Lambda(\sigma_{t-j}^2), \label{eq:augm_GARCH_pq_2}
\end{align}
where $(\epsilon_t)$ is a series of iid rv's with mean $0$ and variance $1$, $\sigma_t^2 = \Var(X_t)$ and $\Lambda, g_i, c_j, i=1,...,p, j=1,...,q$, are real-valued measurable functions. 
Also, as in \cite{Lee14}, we restrict the choice of $\Lambda$ to the so-called group of either polynomial GARCH($p$, $q$) or exponential GARCH($p$, $q$) processes:
\[ (Lee) \quad \quad \Lambda(x) = x^{\delta}, \text{~for some~} \delta >0, \quad \text{~or~} \quad \Lambda(x) = \log(x).\]
Clearly, for a strictly stationary solution to~\eqref{eq:augm_GARCH_pq_1} and~\eqref{eq:augm_GARCH_pq_2} to exist, the functions $\Lambda, g_i, c_j$ as well as the innovation process $(\epsilon_t)_{t \in \Z}$ have to fulfill some regularity conditions (see e.g. \cite{Lee14}, Lemma 1).
Alike, for the bivariate FCLT to hold, certain conditions need to be fulfilled; we list them in the following.
\\ First, conditions concerning the dependence structure of the process $X$. We use the concept of $L_p$-near-epoch dependence ($L_p$-NED), 
using a definition due to Andrews in \cite{Andrews88} 
but restricted to stationary processes. 
Let $(Z_n)_{n \in \Z}$, be a sequence of rv's and $\mathcal{F}_s^t =  \sigma(Z_s,...,Z_t)$, for $s \leq t$, the corresponding sigma-algebra. By $\lvert \cdot \rvert$ we denote the euclidean norm and the usual $L_p$-norm is denoted by $\| \cdot \|_p := \E^{1/p}[ \lvert \cdot \rvert^p]$.
Let us recall the $L_p$-NED definition.
\begin{definition}[$L_p$-NED, \cite{Andrews88}] \label{def:strong_mixing}
For $p>0$, a stationary sequence $(X_n)_{n \in \Z}$ is called $L_p$-NED on $\left( Z_n\right)_{ n\in \Z }$ if for $k \geq 0$
\[ \| X_1 - \E[X_1 \vert \mathcal{F}_{n-k}^{n+k}] \|_p \leq \nu(k), \]
for non-negative constants $\nu(k)$ such that $\nu(k) \rightarrow 0$ as $k \rightarrow \infty$.

If $\nu(k)= O(k^{-\tau -\epsilon})$ for some $\epsilon >0$, we say that $X_n$ is $L_p$-NED of size $\left(-\tau\right)$. \\ If $\nu(k)=O(e^{-\delta k})$ for some $\delta>0$, we say that $X_n$ is geometrically $L_p$-NED.
\end{definition}

The second set of conditions concerns the distribution of the augmented GARCH($p$, $q$) process. 
We impose three different types of conditions as in the iid case (see \cite{Brautigam19_iid}):
First, the existence of a finite $2k$-th moment for any integer $k>0$ for the innovation process $(\epsilon_t)$.
Then, given that the process $X$ is stationary, the continuity or $l$-fold differentiability of its distribution function $F_X$ (at a given point or neighbourhood) for any integer $l> 0$, 
and the positivity of its density $f_X$ (at a given point or neighbourhood). Those conditions are named as:
\begin{align*}
&(M_k) &&\E[\lvert \epsilon_0 \rvert^{2k}] < \infty, 
\\ &(C_0) && F_X \text{~is continuous}, 
\\ \phantom{text to make the distance less} &(C_l^{~'}) &&F_X \text{~is~} l\text{-times differentiable,} \phantom{text to make the distance between \& and \& \& less} 
\\ &(P) &&f_X \text{~is positive.} 
\end{align*}
The third type of conditions is set on the functions $g_i, c_j, i=1,...,p, j=1,...,q$ of the augmented GARCH($p$, $q$) process of the $(Lee)$ family: Positivity of the functions used and boundedness in $L_r$-norm for either the polynomial GARCH, $(P_r)$, or exponential/logarithmic GARCH, $(L_r)$, respectively, for a given integer $r>0$,\\
\begin{align*}
&(A) && g_i \geq 0, c_j \geq 0, i=1,...,p,~j=1,...,q,
\\ &(P_r) && \sum_{i=1}^p \| g_i(\epsilon_0) \|_r < \infty, \quad \sum_{j=1}^q \| c_j(\epsilon_0) \|_r < 1, 
\\ &(L_r) &&  \E[ \exp(4r \sum_{i=1}^p \lvert g_i(\epsilon_0) \rvert^2)] < \infty, \quad \sum_{j=1}^q \lvert c_j(\epsilon_0) \rvert < 1. 
\end{align*}
Note that condition $(L_r)$ requires the $c_j$ to be bounded functions.

\begin{remark}
By construction from \eqref{eq:augm_GARCH_pq_1} and~\eqref{eq:augm_GARCH_pq_2} $\sigma_t$ and $\epsilon_t$ are independent (and $\sigma_t$ a functional of $(\epsilon_{t-j})_{j=1}^{\infty}$). Thus, the conditions on the moments, distribution and density could be formulated in terms of $\epsilon_t$ only. At the same time this might impose some conditions on the functions $g_i, c_j, i=1,...,p, j=1,...,q$ (which might not be covered by $(A)$, $(P_r)$ or $(L_r)$). Thus, we keep the conditions on $X_t$ even if they might not be minimal.
\end{remark}

{\sf \bf \large Risk Measures}

Finally, let us recall the definitions of the risk measures we consider in this paper.
One of the most used risk measures, Value-at-Risk (VaR), is simply a quantile at a certain level of the underlying distribution. 
The VaR for risk management was popularised by JP Morgan in 1996 (see \cite{JPMorgan96}) and is defined as follows:
If we assume a loss random variable $L$ having a continuous, strictly increasing distribution function $F_L$, the VaR at level $\alpha$ of $L$ is simply the quantile of order $\alpha$ of $L$:
\begin{equation}\label{eq:VaRDef}
\VaR_{\alpha} (L) = \inf \Big\{ x :  P [L \leq x] \geq \alpha \Big\} = F_L^{-1}(\alpha).
\end{equation}
Despite the availability of other approaches, the VaR is in practice usually still estimated on historical data (see e.g. \cite{Perignon2010} or \cite{EBA2019} for quantitative surveys on this matter), using the empirical quantile $\widehat{\VaR}_{n} (\alpha) =q_{n}(\alpha)$ associated to a $n$-loss sample $(L_{1}, \dotsc, L_{n})$  with $\alpha \in (0,1)$.

VaR has been shown not to be a coherent measure, \cite{Artzner99}, contrary to Expected Shortfall (ES), introduced in slightly different formulations in~\cite{Artzner97}, \cite{Artzner99}, \cite{Acerbi02}, \cite{Rockafellar02}. ES is defined as follows (e.g. \cite{Acerbi02}) for a loss random variable $L$ and a level $p \in (0,1)$ :
\begin{equation} \label{eq:ES_def}
\ES_{p}(L) = \frac{1}{1-p} \int_{p}^1 q_L (u) du =\E[L \vert L \geq q_L(p)]. 
\end{equation}
While the first equality in \eqref{eq:ES_def} is the definition of ES, the second one holds only if $L$ is continuous. 
There are different ways of estimating ES, we focus on the two most direct ones when using historical estimation. 

First, simply approximating the conditional expectation in \eqref{eq:ES_def} by averaging over $k$ sample quantiles, i.e. 
\begin{equation} \label{eq:ES_nk}
\widetilde{\ES}_{n,k} (p):= \frac{1}{k} \sum_{i=1}^k q_n(p_i),
\end{equation} 
for a specific choice of $p=p_1 < p_2 <...<p_k <1$. This was e.g. proposed in \cite{Emmer15} in the context of backtesting expected shortfall (using $p_i = 0.25~p (5-i)+ 0.25(i-1),~i=1,...,4$). 
Another way was proposed in \cite{Chen08} as
\begin{equation} \label{eq:ES_estim_Chen}
\widehat{\ES}_{n} (p):= \frac{1}{n-[np]+1} \sum_{i=1}^n L_i  ~\1_{(L_i \geq q_n (p))}.
\end{equation}
It can be seen as a special case of $\widetilde{\ES}_{n,k} (p)$ choosing $k=n-[np]+1$ and the $p_i$ accordingly.

The discussions about which risk measure would be most appropriate to use for evaluating the risk of financial institutions have often included a third risk measure, the expectile. It was introduced, in the context of least-squares estimation in \cite{Newey87} and then as a risk measure in \cite{Kuan09}. 
This risk measure satisfies many favourable properties (in particular for backtesting), making it appealing from a theoretical point of view (see e.g. \cite{Bellini14},  \cite{Bellini17} and references therein) but not (yet?) in practice (see e.g. \cite{Emmer15}). It is defined, for a square-integrable loss random variable $L$ and level $p \in (0,1)$, by the following minimiser
\begin{equation} \label{eq:def_expectile}
 e_{p} (L) = \argmin_{x \in \R} \left( p\E[\max(L-x,0)^2] + (1-p) \E[\max(x-L,0)^2] \right).
\end{equation}
While a natural estimator for the expectile is the empirical argmax of \eqref{eq:def_expectile}, 
there exists another way to define an estimator of $e_p$. Recall the relation between an expectile and quantile, see \cite{Yao96}: Let $q_L(p)$ be the quantile at level $p \in (0,1)$, then there exists a bijection $\kappa: (0,1) \mapsto (0,1)$ such that $e_{\kappa(p)} (L) = q_L(p)$ with
\begin{equation} \label{eq:kappa_for_expectile_def}
\kappa ( p ) = \frac{p q_L(p) - \int_{-\infty}^{q_L(p)} x dF_L(x)}{\E[L] - 2 \int_{-\infty}^{q_L(p)} x dF_L(x) - (1-2p) q_L(p)}.
\end{equation}
Thus, such a sample estimator for the expectile at level $p$, exploiting this relation, is denoted as 
\begin{equation} \label{eq:expectile_estim}
e_n(p) := q_n (\kappa^{-1} (p)).
\end{equation}

As unified notation, representing these risk measures, and their estimators, we introduce, for $i=1,...,4$: 
\begin{equation} \label{eq:unified_rm}
\zeta_i(p) = \begin{cases} \Var_{p} (L) & \text{~for~} i=1, \\  \ES_{p} (L) & \text{~for~} i=2, \\ \sum_{i=1}^k  \Var_{p_i} (L)/k  & \text{~for~} i=3,\\ e_{p}(L) & \text{~for~} i=4. \end{cases} \quad \text{~with estimators~} \zeta_{n,i}(p) = \begin{cases} \widehat{\VaR}_{n}(p) & \text{~for~} i=1, \\  \widehat{\ES}_{n}(p) & \text{~for~} i=2, \\ \widetilde{\ES}_{n,k}(p) & \text{~for~} i=3,\\ e_{n}(p) & \text{~for~} i=4. \end{cases}
\end{equation}

{\sf \bf \large Setup of Statistical Framework}

Lastly, we comment on the statistical framework needed to assess the pro-cyclicality. 
Following the empirical study developed in \cite{Brautigam19_emp}, the measure of interest is the linear correlation of the logarithm of a ratio of sample quantiles with the sample MAD ($\hat{\theta}_n)$, namely 
\begin{equation} \label{eq:BDK18}
\Cor\left(\log\left \lvert \frac{\widehat{\VaR}_{n,t+1y}(p)}{\widehat{\VaR}_{n,t}(p)} \right\rvert, \hat{\theta}_{n,t}\right). 
\end{equation}
Here we extend this setup to a more general choice of dispersion measure and risk measure estimators. 
As measure of dispersion estimators, we consider the r-th absolute central sample moment, and as risk measures the ones presented in~\eqref{eq:unified_rm}. 
For this, we need to introduce a time-series notation of our estimated quantities:
Thus, by $\widehat{\VaR}_{n,t}, \widehat{\ES}_{n,t}, \widetilde{\ES}_{n,k,t}, e_{n,t}, \zeta_{n,i,t}, \hat{m}(X,n,r,t)$ we denote, corresponding estimators estimated at time $t$ over the last $n$ observations before time $t$. 

Above all, we are interested in the correlation of the asymptotic distribution corresponding to \eqref{eq:BDK18}. 
Note that by the choice of the sample size $n$ of $n=252$ (in the empirical study of \cite{Brautigam19_emp}) in \eqref{eq:BDK18}, the quantile estimator $\widehat{\VaR}_{n,t+1y}(p)$ is computed on disjoint samples with respect to the other two estimators, i.e. $\widehat{\VaR}_{n,t}(p)$ and $\hat{\theta}_(n,t)$.

Thus, some care has to be taken to translate the setting of \eqref{eq:BDK18} into an asymptotic one (where we let $n \rightarrow \infty$).
For the asymptotic framework at a fixed time $t$, consider a sample of overall size $n$. Then, the trick to have the disjointness of estimators, as in \eqref{eq:BDK18}, is to consider
$\widehat{\VaR}_{n/2, \, t+n/2}(p), \widehat{\VaR}_{n/2,t}(p)$ and $\hat{\theta}_{n/2,t}$, where we assume wlog that $n/2$ is an integer. It means that the VaR and MAD estimators are estimated on a sample of size $n/2$ each. 

More generally, we are interested in the joint asymptotic distribution of the log-ratio, i.e. 
\\$\log\left \lvert \frac{{\zeta}_{n/2, \, t+n/2, \, i}(p)}{{\zeta}_{n/2, \,t, \,i}(p)} \right\rvert$, with the r-th absolute central sample moment $\hat{m}(X,n/2, \, r, \,t)$. 

Then, the generalized analogue to \eqref{eq:BDK18}, i.e. the correlation of the {\it asymptotic distribution} of (these) two quantities, is denoted, to ease and by abuse of notation, as
\begin{equation} \label{eq:procycl-cor-general}
\lim_{n \rightarrow \infty} \Cor\left(\log\left \lvert \frac{{\zeta}_{n/2, \, t+n/2, \, i}(p)}{{\zeta}_{n/2, \,t, \,i}(p)} \right\rvert, \hat{m}(X,n/2, \, r, \,t)\right),
\end{equation}
for $i=1,...,4$, and any integer $r>0$. 
Consequently, our measure of the pro-cyclicality of risk measure estimators amounts to the degree of negative correlation of \eqref{eq:procycl-cor-general}.

A more formal treatment of this will be given in the proofs of Theorem~\ref{thm:procycl_iid_formal} (iid case) and Theorem~\ref{thm:procycl_GARCH_formal} (augmented GARCH($p$, $q$) processes).

\section{Results on Pro-cyclicality} \label{sec:procycl}
The aim of this section is to theoretically assess the pro-cyclicality (of risk measure estimators), i.e \eqref{eq:procycl-cor-general}, in iid models as well as for augmented GARCH($p$, $q$) models. For this, we establish the joint asymptotics between the log-ratio of risk measure estimators and the r-th absolute centred sample moment estimators.

Such results are based on the bivariate CLT's between the risk measure estimators themselves and the r-th absolute central sample moment. For iid models they can be found in the Appendix~\ref{sec:iid_asymptotics} and correspondingly, for augmented GARCH($p$, $q$) processes, in the Appendix~\ref{sec:garch_asymptotics}. 

We first consider the pro-cyclicality in iid models in Section~\ref{sec:procycl-iid} and then in Section~\ref{sec:procycl-garch} for augmented GARCH($p$,$q$) processes.

\subsection{Considering IID models} \label{sec:procycl-iid}
Before stating the proposition, let us come back to the informal explanation of pro-cyclicality in the iid case given in the introduction: Recall that for any risk measure estimator at time $t+n/2$, $\zeta_{n/2, \, t+n/2, \, i}(p)$, the sample used is, by construction, disjoint from the sample used at time $t$. Thus the estimator $\zeta_{n/2, t+n/2, \,i}(p)$ will be uncorrelated with the r-th absolute centred sample moment $\hat{m}(X,n/2,r,t)$, at time $t$, as well as with the risk measure estimator $\zeta_{n/2, \,t, \,i}(p)$ at time $t$.

Translating this for the correlation of the asymptotic distribution (again abusing the notation), i.e.~\eqref{eq:procycl-cor-general}, it should hold, for $i=1,...,4$,
\begin{align*}  
 \lim_{n \rightarrow \infty} &\Cor\left( \log{\lvert \frac{\zeta_{n/2, \, t+n/2, \, i}(p)}{\zeta_{n/2, \,t, \,i}(p)}\rvert}, \,\hat{m}(X,n/2, \, r, \,t) \right) 
 \\ &= \lim_{n \rightarrow \infty} \frac{\Cov( - \log{\lvert \zeta_{n/2, \,t, \,i}(p) \rvert}, \hat{m}(X,n/2, \, r, \,t))} {\sqrt{2 \Var(\log{\lvert \zeta_{n/2, \,t, \,i}(p) \rvert})} \sqrt{\Var(\hat{m}(X,n/2, \, r, \,t))}} 
\\ & = \frac{-1}{\sqrt{2}} \lim_{n \rightarrow \infty} \Cor(  \log \lvert \zeta_{n,\, t, \,i}(p) \rvert, \hat{m}(X,n,r,t)) 
= \frac{-1}{\sqrt{2}} \lvert \lim_{n \rightarrow \infty} \Cor(  \zeta_{n,\,t,\,i}(p), \hat{m}(X,n,r,t)) \rvert, \numberthis \label{eq:cor_iid_logratio_rm_general}
\end{align*} 
where the first equality follows by the uncorrelatedness, the second by the scale invariance of the correlation and the third is a consequence of the Delta-method with the logarithm.
But, anticipating the more involved formal treatment needed for augmented GARCH($p$, $q$) processes, we also present the result in the iid case in a precise way.

\begin{theorem} \label{thm:procycl_iid_formal}
Consider a risk measure estimator $\zeta_{n,i}$, $i \in \{ 1,...,4\}$, and the r-th absolute central sample moment $\hat{m}(X,n,r)$, for a chosen integer $r>0$. Asumme that the conditions for a bivariate FCLT between these estimators are fulfilled (Theorem~\ref{th:qn-abs-central-moment} or Proposition~\ref{prop:ES-abs-central-moment} respectively). 

Then, the asymptotic distribution of the logarithm of the look-forward ratio of the risk measure estimator with the r-th absolute central sample moment is bivariate normal too, i.e.
\[ \sqrt{n} \begin{pmatrix} \log \left\lvert \frac{\zeta_{n/2, \, t+n/2, \, i}(p)}{\zeta_{n/2, \,t, \,i}(p)} \right \rvert  \\ \hat{m}(X,n/2, \, r, \,t) - m(X,r) \end{pmatrix} \overset{d}{\rightarrow} \mathcal{N}(0, \tilde{\Gamma}), \]
and it holds that $\tilde{\Gamma}_{jk} = \begin{cases} \Gamma_{jk}/\zeta_i^2(p) & \text{~for~} j=k=1,
\\  \Gamma_{jk}/2 & \text{~for~} j=k=2  ,
\\ -\Gamma_{jk}/\zeta_i(p)   & \text{otherwise.}  \end{cases} $
In particular, the correlation of this asymptotic bivariate distribution equals 
\[ \frac{\tilde{\Gamma}_{12}}{\sqrt{\tilde{\Gamma}_{11}} \sqrt{\tilde{\Gamma}_{22}}} = \frac{-1}{\sqrt{2}} \sgn(\zeta_i(p)) \frac{\Gamma_{12}}{\sqrt{\Gamma_{11}} \sqrt{\Gamma_{22}}} = \frac{-1}{\sqrt{2}} \frac{\lvert \Gamma_{12} \rvert}{\sqrt{\Gamma_{11}} \sqrt{\Gamma_{22}}}, \]
where $\Gamma$ is the covariance matrix of the asymptotic bivariate distribution between $\zeta_{n,i}$ and $\hat{m}(X,n,r)$. 
\end{theorem}

\subsection{Considering augmented GARCH($p$, $q$) models} \label{sec:procycl-garch}
As second model, we turn now to assessing the pro-cyclicality for the family of augmented GARCH($p$,$q$) processes.

As those processes exhibit dependence, the two estimators, even if computed over disjoint samples, might be correlated (in contrast to the iid case). But it turns out that in our specific case the condition of strong mixing with geometric rate will make the estimators on disjoint samples asymptotically uncorrelated.
Thus, we recover, structurally, the pro-cyclicality behaviour as in the iid case (recall our informal reasoning, ~\eqref{eq:cor_iid_logratio_rm_general}~).
Let us now state, as a theorem, the analogous result to~Theorem~\ref{thm:procycl_iid_formal}.

\begin{theorem} \label{thm:procycl_GARCH_formal}
Consider an augmented GARCH($p$, $q$) process $X$ as defined in~\eqref{eq:augm_GARCH_pq_1} and~\eqref{eq:augm_GARCH_pq_2}, a risk measure estimator $\zeta_{n,i}$, $i \in \{ 1,...,4\}$, and the r-th absolute central sample moment $\hat{m}(X,n,r)$, for a given integer $r>0$. Asumme that the conditions for a bivariate FCLT between these estimators are fulfilled (Theorem~\ref{th:qn-abs-central-moment_garch} or Proposition~\ref{prop:ES-abs-central-moment-garch}, respectively). 

If, moreover, $X$ \textbf{~is strongly mixing with geometric rate and additionally $(M_{r+ \delta})$ holds for some $\delta>0$, }
the asymptotic distribution of the logarithm of the look-forward ratio of the risk measure estimators with the r-th absolute central sample moment is bivariate normal too, i.e.
\[ \sqrt{n} \begin{pmatrix} \log \left\lvert \frac{\zeta_{n/2, \, t+n/2, \, i}(p)}{\zeta_{n/2, \,t, \,i}(p)} \right \rvert  \\ \hat{m}(X,n/2, \, r, \,t) - m(X,r) \end{pmatrix} \overset{d}{\rightarrow} \mathcal{N}(0, \tilde{\Gamma}), \]
and it holds that $\tilde{\Gamma} = \begin{cases} \Gamma_{jk}/\zeta_i^2 (p) & \text{~for~} j=k=1,
\\ \Gamma_{jk}/2 & \text{~for~} j=k=2,
\\ -\Gamma_{jk}/\zeta_i (p)  & \text{otherwise.}  \end{cases}$
In particular, the correlation of this asymptotic bivariate distribution equals to 
\[ \frac{\tilde{\Gamma}_{12}}{\sqrt{\tilde{\Gamma}_{11}} \sqrt{\tilde{\Gamma}_{22}}} = \frac{-1}{\sqrt{2}} \sgn(\zeta_i (p)) \frac{\Gamma_{12}}{\sqrt{\Gamma_{11}} \sqrt{\Gamma_{22}}} = \frac{-1}{\sqrt{2}} \frac{\lvert \Gamma_{12} \rvert}{\sqrt{\Gamma_{11}} \sqrt{\Gamma_{22}}}, \]
\end{theorem}
where $\Gamma$ is the covariance matrix of the asymptotic distribution between $\zeta_{n,i}(p)$ and $\hat{m}(X,n,r)$. 


\begin{remark}
Let us comment on the two additional conditions, with respect to those of Theorem~\ref{th:qn-abs-central-moment_garch}, introduced in the Theorem~\ref{thm:procycl_GARCH_formal}, namely the strong mixing with geometric rate and $(M_{r + \delta})$. We need this dependence condition to make sure that the estimators we consider are asymptotically uncorrelated when computed over disjoint samples.
The moment condition $(M_{r+\delta})$ comes from the fact that we use a CLT for non-stationary, strong mixing processes (\cite{Politis97}, \cite{Ekstrom14}), which requires a stronger condition than the classical $(M_r)$. 

Contrary to these observations, recall that in Proposition~\ref{prop:ES-abs-central-moment-garch}, when establishing a FCLT with $\widehat{\ES}_n (p)$, we needed the strong mixing condition. 
Thus, for the estimator $\widehat{\ES}_n (p)$, Theorem~\ref{thm:procycl_GARCH_formal} imposes only in the case $r>1$ stronger conditions, namely a stronger moment condition, $(M_{r+\delta})$, needed to prove the pro-cyclicality.
%
%
\end{remark}

\section{Application} \label{sec:exp}
In this section we consider two different applications of the theoretical results established on the pro-cyclicality of risk measures in Section~\ref{sec:procycl}.

First, in Section~\ref{ssec:exp-iid}, we want to assess the pro-cyclicality, i.e.~\eqref{eq:procycl-cor-general}, explicitly. This means to compute and compare the pro-cyclicality of five risk measure estimators ($\VaR_{n}(p), e_n (p), \widetilde{\ES}_{n,k}(p)$ for $k=4$,$k=50$ and $k=\infty$) with the two most used central absolute sample moments, the sample MAD ($\hat{m}(X,n,1)$) and the sample variance ($\hat{m}(X,n,2)$).
In contrast to models from augmented GARCH($p$, $q$) processes, for the iid case the closed form expressions of~\eqref{eq:procycl-cor-general} can be computed. 
Thus, we only consider the latter case and look at, as two exemplary distributions, the Gaussian distribution and a Student distribution with varying degrees of freedom. 
This way we can compare how the degree of pro-cyclicality varies for different choices of risk measures, dispersion measures and underlying distributions.

As a second application, we use the result on these theoretical pro-cyclicalities for the two models (Section~\ref{sec:procycl-iid} and~\ref{sec:procycl-garch}), to see if we can add evidence to the empirical claims on the pro-cyclicality of real data in \cite{Brautigam19_emp}. 
Recall that therein it was claimed that part of the pro-cyclicality in the real data should be due to the GARCH effects (as the pro-cyclicality of simulated GARCH($1, 1$) values was similar to the one in the real data), while the other part should be due to the very way risk is estimated (as observed in the iid case). 
From Theorem~\ref{thm:procycl_GARCH_formal}, we know that pro-cyclicality in augmented GARCH($p$, $q$) processes is not an artificial artefact. Still, we cannot use the results to compute the theoretical pro-cyclicality for such processes (and compare it with the one on real data).

Instead, we consider the residuals of the GARCH($1$,$1$) process fitted to the data from \cite{Brautigam19_emp}.
If the pro-cyclicality in the data is due to the GARCH effects, the pro-cyclical behaviour of these residuals should be as the one from iid samples.
With this procedure we provide an additional, alternative argumentation why the pro-cyclicality effects in the data are related partly to an intrinsic part (as observed in iid models) and partly to the volatility behaviour represented by a GARCH($1$,$1$) model.

\subsection{Comparing pro-cyclicality in IID models} \label{ssec:exp-iid}
In the following, we consider the pro-cyclicality as in \eqref{eq:procycl-cor-general} (i.e. the correlation in the asymptotic distribution of the log-ratio of risk measure estimators with measure of dispersion estimators) for underlying iid models.
We consider as risk measure estimator one VaR estimator ($\widehat{\VaR}_n (p)$), one expectile estimator ($e_n(p)$) and three ES estimators ($\widetilde{\ES}_{n,4}, \widetilde{\ES}_{n,50}, \widetilde{\ES}_{n, \infty}$). As measure of dispersion estimator, we focus on the sample MAD ($\hat{m}(X,n,1)$) and the sample variance ($\hat{m}(X,n,2)$).
The closed form solutions follow from Theorem~\ref{thm:procycl_iid_formal} and the corresponding bivariate CLT's, and can be found in Appendix~\ref{appendix:formulas}.
Here we focus on plotting and comparing them.

We start by presenting the results for the Gaussian distribution, ${\cal N} (0,1)$, and then the Student-t distributions with $\nu$ degrees of freedom, choosing $\nu=3,4,5,10$ or $50$ but always normalized to have mean $0$ and variance $1$.

\paragraph{Gaussian Distribution}

In Figure~\ref{fig:cor_risk_measures_norm} we plot the correlations in the asymptotic distribution of the different risk measure estimators with the sample variance (left column) and the sample MAD (right column), respectively. In the second row, we zoom into the tail as, from a risk management point of view, we are interested in the behaviour for high values of $p$.

Looking at the plots in the first row, we see that we have the same tendencies of the correlation of the asymptotic distribution (for VaR, ES and expectile respectively), irrespectively of the choice of the dispersion measure (left plot with the variance, right with the MAD).
Let us take a closer look at the correlation of the asymptotic distribution with the sample variance. VaR and expectile have a similar behaviour, being symmetric around $p=0.5$ (where the correlation equals zero), then increasing to a maximum (in absolute values) and for tail values again, decreasing in direction of $0$ correlation.
The ES, being an integral/sum over the VaR, is not symmetric around $p=0.5$. The location of its zero depends on the estimation method. The correlation increases (in absolute value) from its zero on, until it reaches its maximum for an upper tail value of $p$, then decreases again when p tends to 1.
Further, we see that $\widetilde{\ES}_{n,4}$ is quite different from $\widetilde{\ES}_{n, \infty}$, while $\widetilde{\ES}_{n,50}$ approximates the latter already well.
For $p\geq 0.5$, the ES has clearly higher correlation of the asymptotic distribution than the VaR (except in the tail where they are quite similar).
The correlation of the asymptotic distribution of the expectile is lower than with VaR and ES, except in the tail where it is highest.
For the MAD in the right plot, the same observations hold, only that the maximum value of correlation decreases (slightly) and the location of these maxima is further away from the boundary values of $p$ (especially for the ES estimators).

Looking at the second row of Figure~\ref{fig:cor_risk_measures_norm}, we see a zoom of the correlation plots for high values of the quantile level ($p>0.8$).
In the case with the sample variance, we see that for values of $p<0.97$, in absolute values, the correlation with the expectile is lowest while the one with the ES (irrespective of the choice of estimator) is the highest. For values further in the tail, the behaviour is inverted and the correlation with the ES and VaR are very similar. 
Further, all correlations seem to tend to $0$ for $p\rightarrow1$.
On the right plot, in the case of the MAD as dispersion measure, we see the same behaviour, only that the threshold at which the behaviour is inverted is already at $p=0.92$.

\begin{figure}[H]	
%
\begin{minipage}{0.5\textwidth}
\includegraphics[scale=1.05,width=6cm,height=6cm]{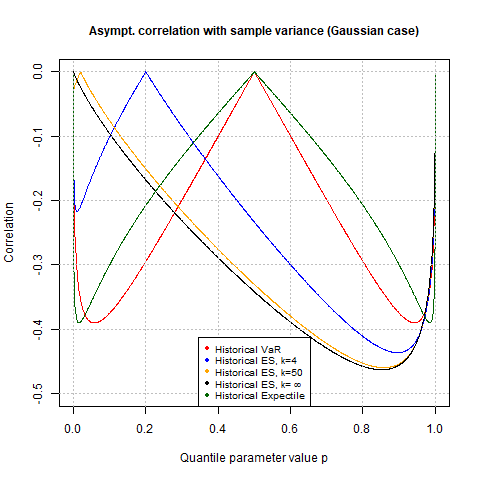}
\\ \includegraphics[scale=1.05,width=6cm,height=6cm]{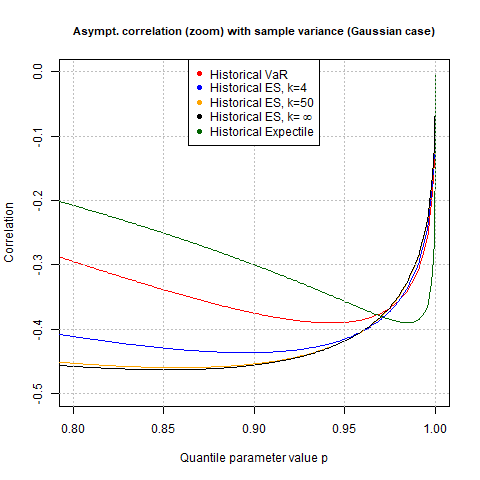}
\end{minipage}
\begin{minipage}{0.5\textwidth}
\includegraphics[scale=1.05,width=6cm,height=6cm]{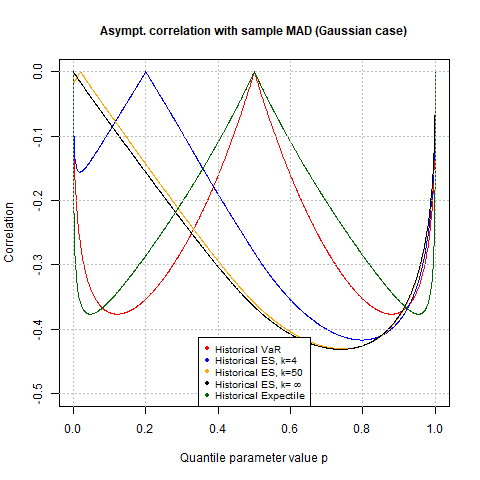}
\\ \includegraphics[scale=1.05,width=6cm,height=6cm]{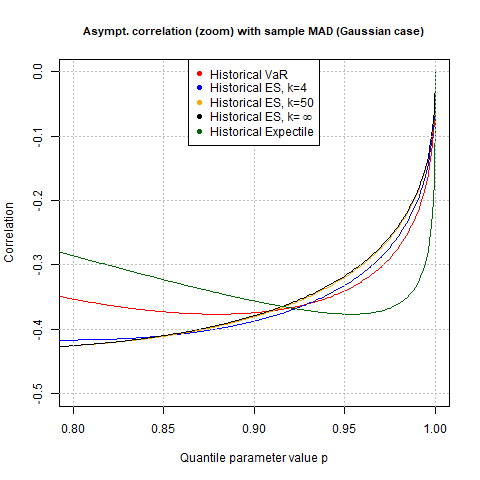}
\end{minipage}
\caption{\label{fig:cor_risk_measures_norm} \sf\small Pro-cyclicality as defined in~\eqref{eq:procycl-cor-general}, considering on each plot three different risk measures (VaR, ES, evaluated in 3 possible ways, and expectile). In the left column the measure of dispersion is the sample variance, in the right column the sample MAD. Case of an underlying Gaussian distribution.}
\end{figure} 

\paragraph{Student-t Distribution}

We start by considering the case $\nu=5$ in Figure~\ref{fig:cor_risk_measures_stud} since we need $\nu >4$ for $(M_2)$ to hold. 
As the behaviour changes with $\nu$, in a second step, we look in Figure~\ref{fig:cor_riskmeasures_stud_flex_nu} at the correlations as a function of $\nu$ by comparing the cases $\nu=3,4,5,10,40$ with the Gaussian limiting case.

Looking first at the correlation with the sample variance in Figure~\ref{fig:cor_risk_measures_stud} (first column), we see, generally speaking, the same trends as in the Gaussian case. However there are three articulate exceptions to that: 
For $p\geq 0.5$, the correlation with the ES is always higher than with VaR, and with VaR, always higher than with the expectile (in the Gaussian case there was a high threshold for $p$ where this behaviour was inverted). Second, the correlation values with the expectile do not tend to $0$ for p tending to $1$, but rather seem to converge to a non-zero value.
Third, the correlation with $\widetilde{\ES}_{n,50}$ does not apprpoximate the correlation $\widetilde{\ES}_{n, \infty}$, as well as in the Gaussian case.

For the correlation with the MAD (second column of Figure~\ref{fig:cor_risk_measures_stud}), we can say as well that the same trends as in the corresponding Gaussian case are visible. 
But we only share one exception with the case of the sample variance: The correlation of the expectile tends for $p$ tending to $0,1$ to a non-zero value too.

\begin{figure}[H]	
%
\begin{minipage}{0.5\textwidth}
\includegraphics[scale=1.05,width=6cm,height=6cm]{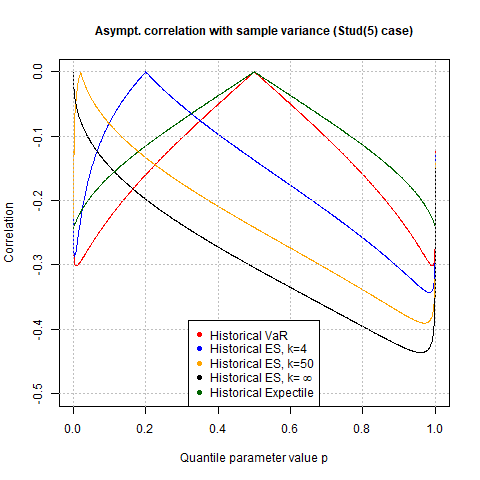}
\\ \includegraphics[scale=1.05,width=6cm,height=6cm]{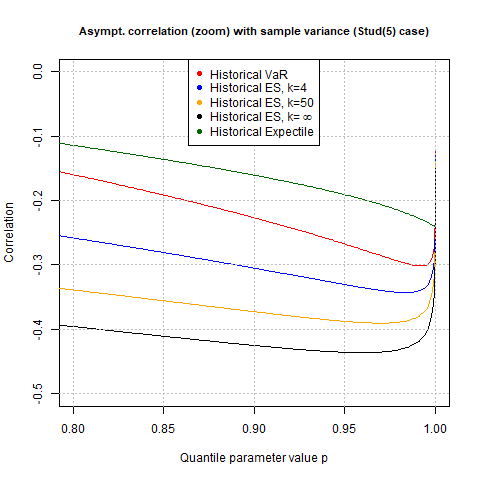}
\end{minipage}
\begin{minipage}{0.5\textwidth}
\includegraphics[scale=1.05,width=6cm,height=6cm]{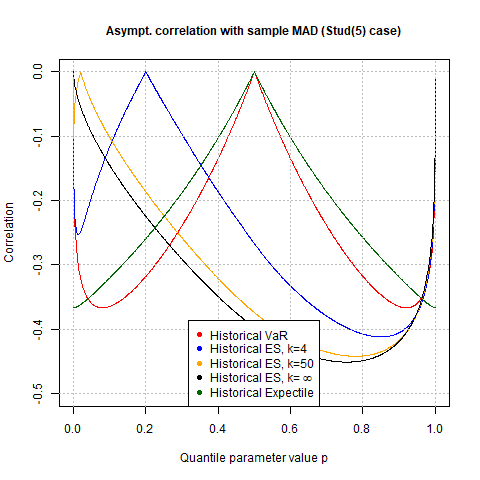}
\\ \includegraphics[scale=1.05,width=6cm,height=6cm]{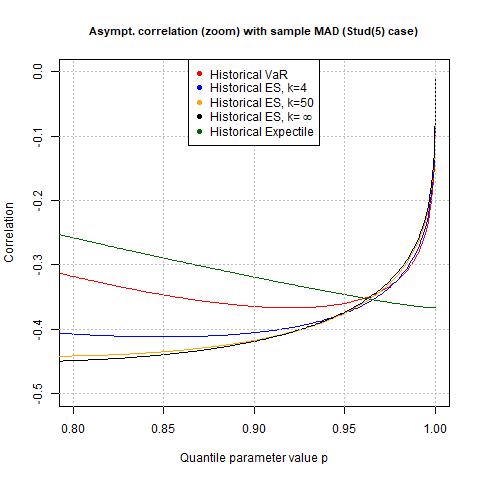}
\end{minipage}
\caption{\label{fig:cor_risk_measures_stud} \sf\small Pro-cyclicality as defined in~\eqref{eq:procycl-cor-general}, considering on each plot three different risk measures (VaR, ES, evaluated in 3 possible ways, and expectile). On each row in each plot a different measure of dispersion is considered (from left to right: sample variance, sample MAD. Case of a Student distribution with 5 degrees of freedom.}
\end{figure}

As mentioned, we also want to study the convergence of the Student correlation to the Gaussian case with respect to the degrees of freedom $\nu$.
Thus, we look in Figure~\ref{fig:cor_riskmeasures_stud_flex_nu} at the correlation for each pair of risk and dispersion measure separately, but showing the cases $\nu=3,4,5,10,40$ and $\infty$ (Gaussian case) in the same plot.

First, we look at the case with the VaR (first row).
For the sample variance (left plot), we see that the convergence, for values near $p=0.5$ is quicker as for the other intermediate values; near the boundaries it seems to behave as near $p=0.5$ but this is difficult to assess from the plot. Further, as we already know for the VaR, the behaviour is symmetric around the $p=0.5$-axis. We observe a similar behaviour with the sample MAD (right plot).  
But we see that the convergence of the correlation for the variance is slower than for the sample MAD. Further, the convergence with the sample MAD is smoother than with the sample variance. E.g. the shape and values from $\nu=5$ to $\nu=10$ change more with the sample variance than with the sample MAD.

Let us now turn to the ES in the second row.
Again, we start with the left plot, i.e. the convergence with the sample variance as measure of dispersion. The behaviour of the correlation changes twice. For rather low values of $p$, the correlation is highest (in absolute terms) for small degrees of freedom, then for intermediate values of $p$ this is inverted, and again for very high values of $p$, we have the same behaviour as for low values of $p$. The speed of convergence varies also with $p$.
In contrast to this, the convergence with the MAD is very uniform. The lower the degree of freedom, the higher the correlation (in absolute terms). The quickest convergence is for values of $p$ between $0.6$ and $0.8$. As we already know, the behaviour of the ES is not symmetric. To the contrary, the convergence for values of $p$ between $0.1$ and $0.3$ is even the slowest.
The expectile (third row) shows the same characteristics as with the VaR, apart from the fact that the convergence for boundary values of $p$ is the slowest for all values of $p$.

\begin{figure}[H]
%
\begin{minipage}{0.5\textwidth}
\includegraphics[scale=1.05,width=6cm,height=6cm]{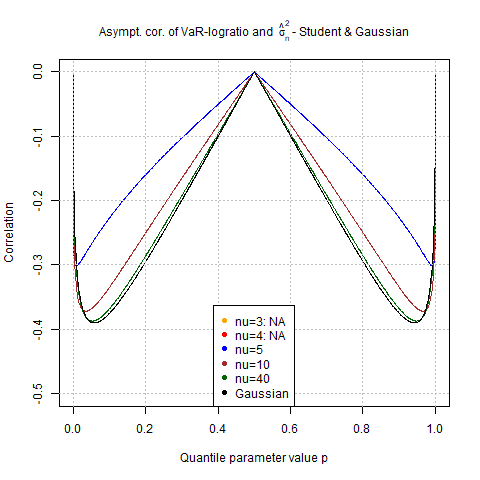}
\\ \includegraphics[scale=1.05,width=6cm,height=6cm]{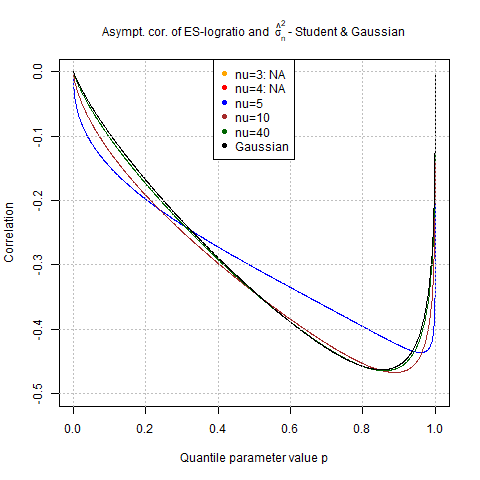}
\\ \includegraphics[scale=1.05,width=6cm,height=6cm]{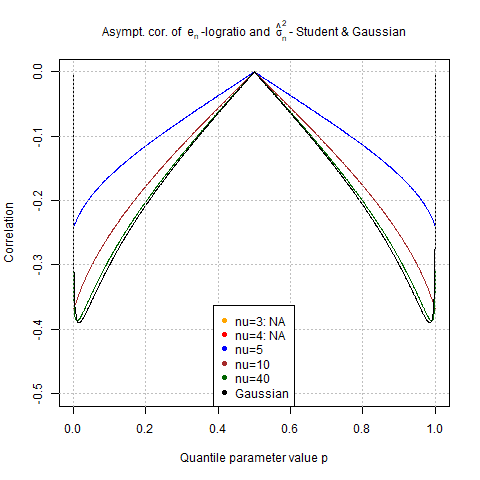}
\end{minipage}
\begin{minipage}{0.5\textwidth}
\includegraphics[scale=1.05,width=6cm,height=6cm]{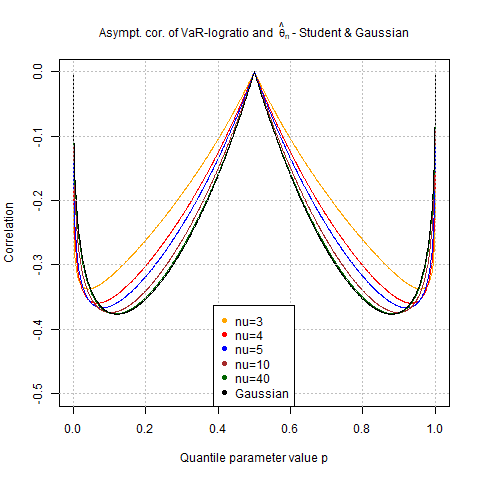}
\\ \includegraphics[scale=1.05,width=6cm,height=6cm]{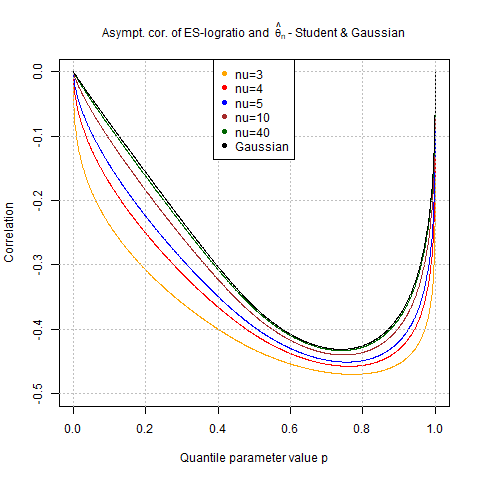}
\\ \includegraphics[scale=1.05,width=6cm,height=6cm]{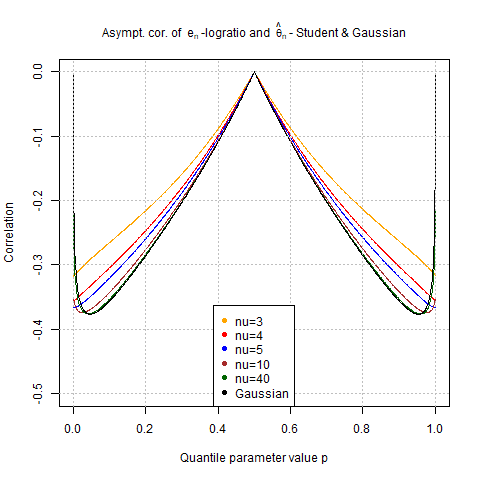}
\end{minipage}
\parbox{470pt}{\caption{\label{fig:cor_riskmeasures_stud_flex_nu} \sf\footnotesize Pro-cyclicality as defined in~\eqref{eq:procycl-cor-general}, comparing the case of a Student distribution -with $\nu= 3,4,5,10$ and $40$ degrees of freedom as well as the Gaussian distribution. From top to bottom: sample quantile, ES, expectile; and from left to right: sample variance, sample MAD.}}
\end{figure}

%

\paragraph{Implications of the pro-cyclicality for the choice of risk measure}
Let us finish the comparison of the pro-cyclicality in Gaussian and Student iid models for the different risk measures by commenting on its implications for the choice of risk measure. 

From the figures we have seen that the pro-cyclicality behaviour depends on the choice of underlying risk measure, dispersion measure and also the distribution. Thus, there is not one simple general tendency to attach to the pro-cyclicality behaviour. Instead, the detailed situation has to be taken into account.

Let us exemplify this in the Gaussian case. {\it If} one is interested in choosing a risk measure which accentuates the pro-cyclical effect most, from the figures we have seen that the expectile would be the measure of choice - but only for high thresholds. In turn, this exact threshold depends on the corresponding measure of dispersion one is using to measure the pro-cyclcality. For the sample variance the expectile had the highest degree of pro-cyclicality for $p>0.97$, whereas with the sample MAD this already holds for $p>0.93$. Below these threshold values the expectile has the lowest degree of pro-cyclicality compared to the other risk measures. Thus, being aware of this threshold value is very important as it might reverse the conclusions! Also, specifically for the expectile,its behaviour is different for heavier tailed distributions. As mentioned, not making it possible, to deduce general tendencies. 	

When being confronted by the choice of ES or VaR (as these risk measures are more common in practice), one can say that one has, in general, more pro-cyclicality with the ES. But then again, this statement has to be quantified. 
This is the case for higher, but non-extreme thresholds $p$. Also, we saw that for heavier tailed distributions this difference was bigger. 
To the contrary we have seen that in the extreme tails VaR exhibits even slightly more pro-cyclicality than the ES (albeit of the same order). 

Thus, to better highlight the effect of pro-cyclicality, the ES is most suited. It has a higher degree of pro-cyclicality than the VaR and in contrast to the expectile its pro-cyclicality behaviour is more consistent. It does not change as drastically (depending on the choice of distribution or measure of dispersion) as the expectile.

\subsection{Pro-cyclicality analysis on real data (reprise)} \label{ssec:residual_analysis}
In this last part, we want to use the thereotical results on pro-cyclicality to address the empirical claims in \cite{Brautigam19_emp}, namely that the pro-cyclicality observed is partly from an intrinsic effect of using historical estimation and partly due to the clustering and return-to-the-mean behaviour of volatility, as modeled with a GARCH($1, 1$).
Thus, it seems logical to use the theoretical results on the pro-cyclicality of augmented GARCH($p$, $q$) processes, Theorem~\ref{thm:procycl_GARCH_formal}, to compute the theoretical value for a GARCH($1, 1$) process and compare it with the value in the real data.

But there are some fallacies to that.
First, for this family of models we do not have closed form solutions of the correlation of the asymptotic distribution.
Further, it is known that, for GARCH processes, the convergence to its asymptotic distribution is slow (as e.g. \cite{Mikosch00} argue for the autocovariance/autocorrelation process).
This means that, contrary to the iid case (as one could see in the simulation study in \cite{Brautigam19_iid}), 
the asymptotic values are not a good approximation for small $n$. 

Thus, we proceed differently in this case. Instead of analysing the theoretical correlation for a GARCH model, we consider the residuals of a GARCH($1, 1$) fitted to the data and analyse the pro-cyclicality of this residual process.
\paragraph{Pro-cyclicality Analysis of Residuals} 
To start with, recall the GARCH(1,1) model: 
\begin{align*}
X_{t+1} &= \epsilon_t\,\sigma_t, 
\\ \text{with~}\sigma_t^2 &= \omega + \alpha \, X_t^2 + \beta \sigma_{t-1}^2 \text{~and~} \omega>0, \alpha \geq 0, \beta \geq 0,
\end{align*}
where $\left( \epsilon_t, t \in \Z \right)$ is an iid series with mean $0$ and variance $1$.

For each of the 11 indices we consider the empirical residuals $\hat{\epsilon}_t := X_{t+1} / \hat{\sigma}_t$. 
Using the GARCH parameters fitted in \cite{Brautigam19_emp}, we initialize $\hat{\sigma}_t$ by using one year of data (as `burn-in' sample).
Then, to assess the pro-cyclicality of the residuals, we compute the sample correlation between the log-ratio of sample quantiles and the sample MAD as in \cite{Brautigam19_emp} - but here, on the time-series of residuals $\hat{\epsilon}_t$ (and not the real data itself!). 
In theory, this time series of residuals should be iid distributed with mean $0$ and variance $1$. 
Hence, using the results of Theorem~\ref{thm:procycl_iid_formal}, we can exactly assess this pro-cyclicality (i.e. the correlation in the asymptotic ditribution of the SQP-logratio and the MAD) of iid models.

To compare the sample correlation (based on a finite sample) with the theoretical asymptotic value of the correlation, we provide the corresponding confidence intervals for the sample Pearson linear correlation coefficient (as done in the iid simulation study, see \cite{Brautigam19_iid} for details).
We iterate that those confidence interval values have to be considered with care. 
They are based on assuming to compute a sample correlation on a bivariate normal sample. But the bivariate normality of the log-ratios of sample quantiles with the sample MAD holds only asymptotically. Hence, it is not clear if, for the sample size considered, we can assume bivariate normality (this could be tested). 
Here, as in the empirical study of \cite{Brautigam19_emp}, we are computing the sample correlation on a sample of size $\sim 300$. From the simulation results (available upon request 
), we can see that, for such a size, the empirical and theoretical confidence intervals for underlying Gaussian and Student samples are similar. Thus, we feel confident in providing those theoretical confidence intervals as approximate guidance.
We then verify if the sample correlation based on the residuals falls in these confidence intervals, and how the sample correlation based on the real data (as computed in \cite{Brautigam19_emp}) behaves in comparison.



\begin{figure}[H]
\centering
\begin{tabular}{ccc}
\subf{\includegraphics[width=49.4mm]{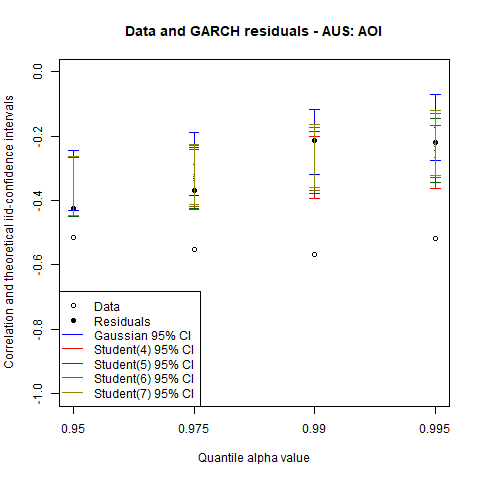}}
     {}
&
\subf{\includegraphics[width=49.4mm]{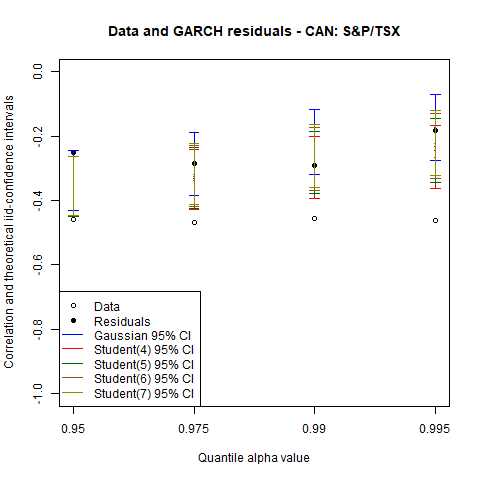}}
     {}
&
\subf{\includegraphics[width=49.4mm]{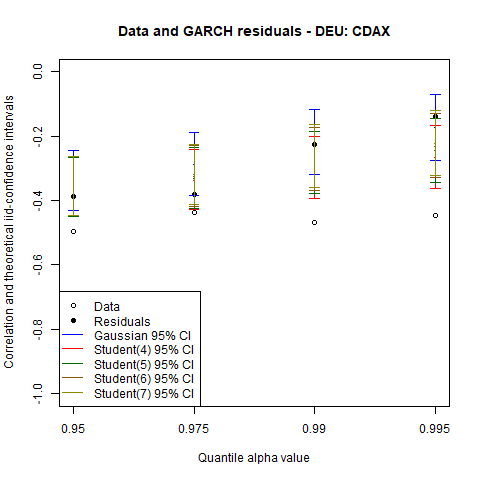}}
     {}
\\
\subf{\includegraphics[width=49.4mm]{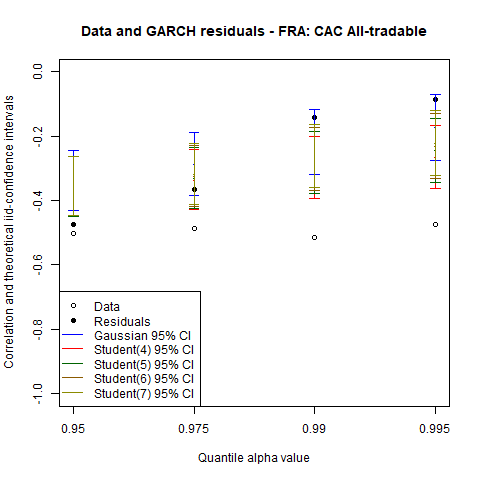}}
     {}
&
\subf{\includegraphics[width=49.4mm]{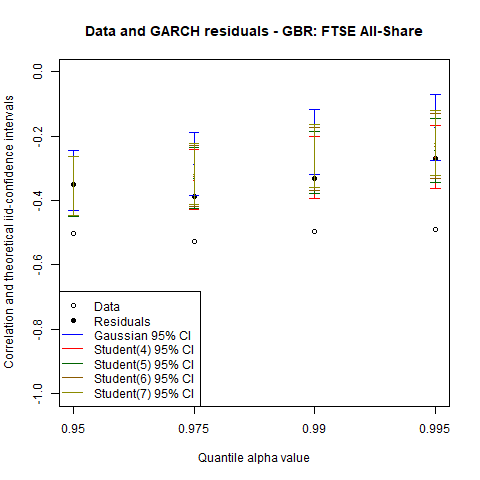}}
     {}
&
\subf{\includegraphics[width=49.4mm]{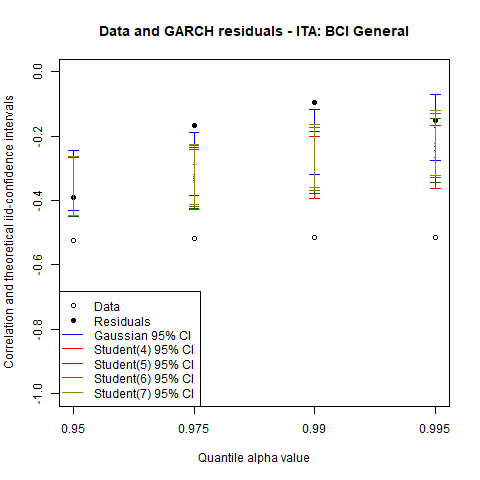}}
     {}
\\
\subf{\includegraphics[width=49.4mm]{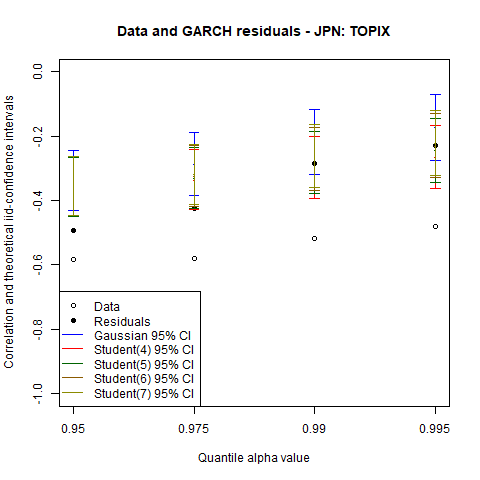}}
     {}
&
\subf{\includegraphics[width=49.4mm]{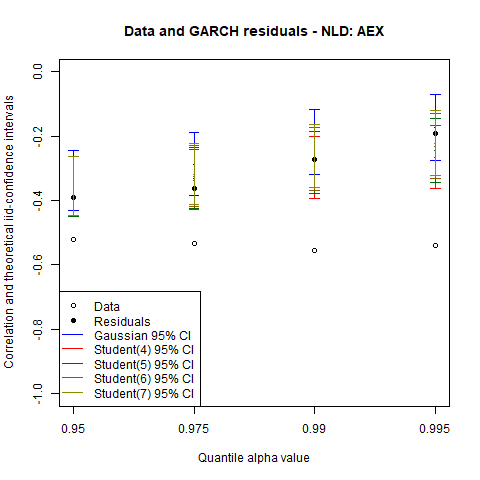}}
     {}
&
\subf{\includegraphics[width=49.4mm]{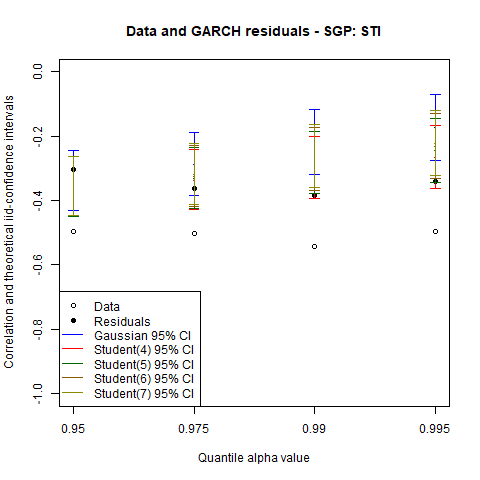}}
     {}
\\
\subf{\includegraphics[width=49.4mm]{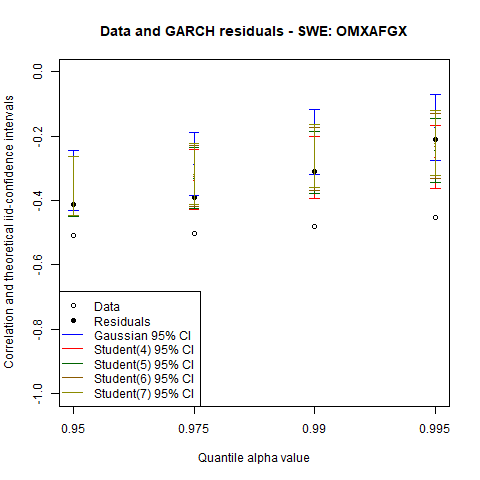}}
     {}
&
\subf{\includegraphics[width=49.4mm]{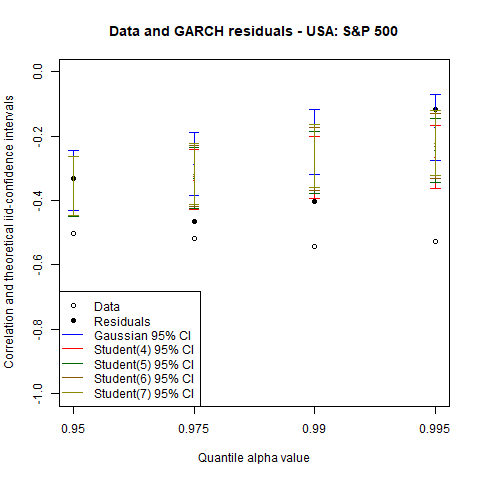}}
     {}
&
\\
\end{tabular}
\caption{\label{fig:cor_residuals_vs_iid}\sf\small Comparison of pro-cyclicality in the data (blank circle) with the pro-cyclicality of the GARCH($1, 1$)-residuals (filled circle) for each index separately.
Each plot contains the correlation for the four different $\alpha$ values. For each of them, corresponding theoretical confidence intervals (for the sample correlation) assuming a specific underlying distribution (Gaussian or Student with different degrees of freedom) are plotted.}
\end{figure}

In Figure~\ref{fig:cor_residuals_vs_iid} we have one plot for each of the 11 indices. In each plot, we compare for each threshold $\alpha = 0.95,0.975,0.99,0.995$, the measured pro-cyclicality (i.e. the sample correlation between the log-ratio of sample quantiles and the sample MAD) on the real data versus the one on the residuals.
Further, 95\%-confidence intervals for a sample correlation assuming an underlying iid distribution are given - considering as alternatives a Gaussian or Student distribution, the latter with varying degrees of freedom, $\nu =4,...,7$.
In 38 out of 44 cases (86\%), the sample correlation of the residuals falls in the 95\% confidence interval of the sample correlation of an iid distribution. But in none of the cases, the sample correlation of the real data falls in these confidence intervals.
Thus, we claim that the pro-cyclical behaviour of the residuals seems to be as the pro-cyclical behaviour of iid random variables.
This finally means that stripping-off the GARCH features of the real data by considering its residuals, we are left with a pro-cyclicality behaviour like for iid data. Hence, the claim of \cite{Brautigam19_emp} has been backed. Namely, that the correlation in the real data is due to two factors: 
One, the inherent pro-cyclicality due to the use of historical estimation as modeled with iid rv's, and a second one due to the GARCH effects, i.e. return-to-the-mean and clustering of volatility.

\section{Conclusion} \label{sec:conclusio}

The goal in this paper was to link the empirical evidence presented in \cite{Brautigam19_emp} with the theoretical results of \cite{Brautigam19_iid},\cite{Brautigam19_garch}.  
In the empirical study, the sample correlation of a log-ratio of sample quantiles with the sample MAD was considered. But the theoretical results of \cite{Brautigam19_iid} and \cite{Brautigam19_garch} treated the (correlation of the) asymptotic distribution between a quantile estimator and a measure of dispersion estimator.
Here, we wanted to assess the pro-cyclicality as measured in \cite{Brautigam19_emp}.

For this, we first needed to define the pro-cyclicality in an asymptotic sense. Also, we extended the setting beyond the VaR as risk measure, also including the ES and expectile (to be able to compare the pro-cyclicality also accross different risk measures). As measure of dispersion, we used the r-th absolute central sample moment.

We then started by tackling the pro-cyclicality in iid models.
While the answer seemed intuitively clear (the risk measure estimators are computed on disjoint iid samples, thus are uncorrelated), we treated this formally:
We considered sequences which are equal to $0$ for half of the sample, and follow the underlying distribution on the other half. In this way, the estimators built on these sequences were uncorrelated.
To compute the desired bivariate asymptotic distribution, we then applied a CLT for independently but non-identically distributed sequences.
Note that to conclude the pro-cyclicality in an iid setting, we needed no extra conditions compared to the bivariate CLT's between the respective risk and measure of dispersion estimators.

Subsequently, we treated the case of augmented GARCH($p$, $q$) processes, establishing analogous results to the iid case.
As additional conditions, we introduced the strong mixing with geometric rate of the underlying process, as well as slightly stronger moment conditions $\Big( (M_{r+\delta})$ instead of $(M_r) \Big)$. 
As in this case the estimators computed on disjoint samples were not any more uncorrelated a priori, we needed the strong mixing with geometric rate to show that we can bound these covariances. We showed that, asymptotically, they are uncorrelated (i.e. asymptotically we recover structurally the same behaviour as in the iid case).

For both types of models considered, we showed the same results 
: Yes, we can mathematically prove the pro-cyclicality (measured by the negative correlation in the asymptotic distribution of the log-ratio of risk measure estimators with the r-th absolute central sample moment).
Further, our results showed that it will be always present, no matter what the choice of model, risk measure (estimator) or measure of dispersion estimator.

As application of these results, we were interested in comparing the pro-cylicality behaviour for different choices of risk and dispersion measure and underlying models. 
We considered the iid model, as we are able to derive closed form solutions in this case. We compared the pro-cyclicality of VaR, ES and expectile with the sample MAD or sample variance, when considering a Gaussian and Student-t distribution with different degrees of freedom.

As last application we examined what we could deduce from these theoretical findings on pro-cyclicality for the  empirically observed pro-cyclicality in \cite{Brautigam19_emp}. As we did not have closed form solutions for the GARCH($1, 1$) case (and the asymptotics do not approximate well the finite sample behaviour), we could not use its theoretical pro-cyclicality results directly.
Instead, we considered an alternative approach:
We assessed the pro-cyclicality of the residual process of the GARCH($1, 1$) fitted to the data as in \cite{Brautigam19_emp}. We showed that in most of the cases (86\%, i.e. 38 out of 44 cases), the pro-cyclicality of the residuals fell into the 95\% confidence bands of the theoretical pro-cyclicality value for Gaussian and Student iid models. In contrast, the pro-cyclicality value of the real data (and not the residuals) did not fall in any of the cases into these confidence bands.
We saw this as an alternative and additional way to support the claim that the pro-cyclicality observed on real data is to one part intrinsically due to the way risk is measured historically and to another part due to the volatility effects as modeled by a GARCH($1, 1$), i.e. the return-to-the-mean and clustering of volatility.

\newpage

\small
\bibliography{qrm}
\bibliographystyle{acm}

\newpage

\section*{APPENDIX}
\vspace{-1ex}
The Appendix consists out of three parts. The first one, \ref{appendix:CLTs}, collects the results (and their proofs) of the (F)CLTs between risk measure estimators and the r-th absolute sample moments. 
The second part, Appendix~\ref{sec:appendix_procycl}, gives the proofs of the pro-cyclicality results of Section~\ref{sec:procycl}.
The third part provides the explicit formuas for the examples computed in Section~\ref{sec:exp}.

\begin{appendices} 



\section{CLT's between risk and dispersion measure estimators} \label{appendix:CLTs}
\subsection{Considering IID models} \label{sec:iid_asymptotics}
We want to establish bivariate CLT's between $\zeta_{n,i}$ and $\hat{m}(X,n,r)$. Note that most cases are already covered by results in \cite{Brautigam19_iid}.

Therein, the asymptotics for the $\widehat{\VaR}_n(p)$ with $\hat{\mu}(X,n,r)$ are given. 
For the sake of completeness, we reiterate the theorem here:

\begin{theorem} \label{th:qn-abs-central-moment} 
Consider an iid sample with parent rv $X$ having existing (unknown) mean $\mu$ and variance $\sigma^2$.
Assume conditions $(C_1^{~'}), (P)$ at $q_X(p)$ each, $(M_r)$ for the correponding integer $r$, as well as $(P)$ at $\mu$ for $r=1$. 
Then the joint behaviour of the functions $h_1$ of the sample quantile $q_n(p)$, for $p \in (0,1)$, and $h_2$ of the r-th sample absolute central moment $\hat{m}(X,n,r)$, is asymptotically normal:
\begin{equation*}
\sqrt{n} \, \begin{pmatrix} h_1(q_n (p)) - h_1(q_X(p)) \\ h_2(\hat{m}(X,n,r))  - h_2(m(X,r)) \end{pmatrix} \; \underset{n\to\infty}{\overset{d}{\longrightarrow}} \; \mathcal{N}(0, \Sigma^{(r)}), 
\end{equation*}
where the asymptotic covariance matrix $\displaystyle \Sigma^{(r)}=(\Sigma^{(r)}_{ij}, 1\le i,j\le 2)$ satisfies 
\begin{align*}
\Sigma^{(r)}_{11}&=\frac{p(1-p)}{f_X^2(q_X(p))} \,\left(h_1'(q_X(p))\right)^2 ; \quad  \Sigma^{(r)}_{22}=\left(h_2'(m(X,r))\right)^2 \, \Var\left(\lvert X - \mu \rvert^r -r (X-\mu) \E[(X-\mu)^{r-1} \sgn(X-\mu)^r] \right) ;  
\\  \Sigma^{(r)}_{12}&= \Sigma^{(r)}_{21} =
h_1'(q_X(p)) \,h_2'(m(X,r)) \times \frac{\Cov (\lvert X - \mu \rvert^r, \1_{(X>q_X(p))}) - r \E[(X-\mu)^{r-1} \sgn(X-\mu)^r] \Cov(X, \1_{(X > q_X(p))}}{f_X(q_X(p))}.  
\end{align*}

The asymptotic correlation between the functional $h_1$ of  the sample quantile and the functional $h_2$ of the r-th absolute sample moment is - up to its sign $a_{\pm} = \sgn( \,h_1'(q_X(p)) \times h_2'(m(X,r)))$ - the same whatever the choice of $h_1,h_2$:
\begin{equation*} 
\frac{\Sigma^{(r)}_{12}}{\sqrt{\Sigma^{(r)}_{11} \Sigma^{(r)}_{22}}} =  a_{\pm} \times \frac{\Cov (\lvert X - \mu \rvert^r, \1_{(X>q_X(p))})  - r \E[(X-\mu)^{r-1} \sgn(X-\mu)^r] \Cov(X, \1_{(X > q_X(p))}}{\sqrt{ p(1-p) \Var\left(\lvert X - \mu \rvert^r -r (X-\mu) \E[(X-\mu)^{r-1} \sgn(X-\mu)^r] \right)}}.
\end{equation*} 
\end{theorem}

As $e_n(p)$, by definition, is a sample quantile at level $\kappa^{-1}(p)$, we can use the same theorem assuming $\kappa$ is given. 
Also, $\widetilde{\ES}_{n,k}(p)$ is, for any finite choice of $k$, an average of $k$ sample quantiles at different levels $p_i, i=1,...,k$.
Thus, its bivariate asymptotics follows from the extension of Theorem~\ref{th:qn-abs-central-moment} to a vector of sample quantiles, Theorem~7 in \cite{Brautigam19_iid} and the continuous mapping theorem.

Thus, only the case of $\widehat{\ES}_{n}(p)$ needs to be considered. 
The approach is the same as in Theorem~\ref{th:qn-abs-central-moment}, only that $\widehat{\Var}_n (p)$ is replaced by $\widehat{\ES}_{n}(p)$, and with it, the conditions required on the underlying distribution slightly change.

\begin{proposition} \label{prop:ES-abs-central-moment}
Consider an iid sample with parent rv $X$ having mean $\mu$, variance $\sigma^2$.
For any integer $r>0$, assume that $(M_r)$ holds, $F_X$ is absolutely continuous, $(C_3)$ holds in a neighbourhood of $q_X(p)$, and, if $r=1$, $(C_0)$ at $\mu$ and $(M_{1+\delta})$ for some $\delta>0$ hold. 
Then the joint asymptotic distribution of the historically estimated expected shortfall $\widehat{\ES}_n (p)$, for $p \in (0,1)$, and the r-th absolute central sample moment $\hat{m}(X,n,r)$, for any integer $r$, is bivariate normal with the following correlation of the asymptotic distribution: $\lim_{n \rightarrow \infty} \Cor( \widehat{\ES}_n (p), \hat{m}(X,n,r)) =$
\begin{equation} \label{eq:cor_ES_var_or_MAD}
\frac{\Cov(\frac{1}{1-p} \left( X -q_X(p) \right) \1_{(X \geq q_X(p))}, \lvert X - \mu \rvert^r - r \E[(X- \mu)^{r-1} \sgn(X- \mu)^r] (X - \mu))}{\sqrt{\Var\left(\lvert X - \mu \rvert^r -r (X-\mu) \E[(X-\mu)^{r-1} \sgn(X-\mu)^r] \right)} \sqrt{\Var(\frac{1}{p} (X -q_X(p)) \1_{(X \geq q_X(p))}}}.
\end{equation} 
\end{proposition}


\begin{remark} \label{rmk:iid-ES-asympt}
Note that the conditions on the underlying distribution are stronger than in the case of the VaR. This comes from the use of the Bahadur representation of the ES estimator. We need absolute continuity of $F_X$ and continuity of the second derivative of $f_X$ in a neighbourhood of $q_X(p)$. In Theorem~\ref{th:qn-abs-central-moment}, we only needed differentiability of $F_X$ and positivity of $f_X$ at the point $q_X(p)$.
Also, in the case of $r=1$, we have an additional moment condition, which comes from the ES estimator, namely the existence of at least the $2+2\delta$th moment.
A thorough examination of the proof in \cite{Chen08} (which is set out for strongly mixing time series) should make it possible to reduce the moment condition to $(M_1)$.
\end{remark}

\subsection{Considering augmented GARCH($p$,$q$) models} \label{sec:garch_asymptotics}
We want to establish FCLT's between $\zeta_{n,i}(p)$, $i=1,...,4$, and $\hat{m}(X,n,r)$ for augmented GARCH($p$, $q$) processes. 
As in the iid case, the bivariate FCLT for the estimator $\widehat{\VaR}_{n}(p)$ was already proven in \cite{Brautigam19_garch}, and we state it for completeness:

To ease its presentation we introduce a trivariate normal random vector (functionals of $X$), $(U, V, W)^T$, with mean zero and the following covariance matrix:
\begin{equation*}
    (D)~\left\{\begin{aligned}
\Var(U) &= \Var(X_0) +2 \sum_{i=1}^{\infty} \Cov(X_i, X_0)
\\ \Var(V) &= \Var(\lvert X_0 \rvert^r) +2 \sum_{i=1}^{\infty}  \Cov(\lvert X_i\rvert^r, \lvert X_0 \rvert^r) 
\\ \Var(W) &= \Var \left( \frac{p- \1_{(X_0 \leq q_X(p))}}{f_X(q_X(p))} \right) + 2 \sum_{i=1}^{\infty} \Cov \left( \frac{p- \1_{(X_i \leq q_X(p))}}{f_X(q_X(p))}, \frac{p- \1_{(X_0 \leq q_X(p))}}{f_X(q_X(p))} \right)
\\ &= \frac{p(1-p)}{f_X^2(q_X(p))} + \frac{2}{f_X^2(q_X(p))} \sum_{i=1}^{\infty} \left( \E[\1_{(X_0 \leq q_X(p))} \1_{(X_i \leq q_X(p))} ] -p^2  \right) \notag
\\ \Cov(U,V) &= \sum_{i \in \mathbb{Z}} \Cov(\lvert X_i \rvert^r,X_0) = \sum_{i \in \Z} \Cov(\lvert X_0  \rvert^r, X_i  )
\\ \Cov(U,W) &= \frac{-1}{f_X(q_X(p))} \sum_{i \in \mathbb{Z}} \Cov( \1_{(X_i \leq q_X(p))},X_0) = \frac{-1}{f_X(q_X(p))} \sum_{i \in \Z} \Cov( \1_{(X_0 \leq q_X(p))},X_i )
\\ \Cov(V,W) &= \frac{-1}{f_X(q_X(p))} \sum_{i \in \mathbb{Z}} \Cov(  \lvert X_0 \rvert^r, \1_{(X_i \leq q_X(p))}) =\frac{-1}{f_X(q_X(p))}  \sum_{i \in \Z} \Cov(\lvert X_i \rvert^r, \1_{(X_0 \leq q_X(p))} ).
 \end{aligned}
    \right.
\end{equation*}
\begin{theorem} \label{th:qn-abs-central-moment_garch} 
For an integer $r>0$, consider an augmented GARCH($p$, $q$) process $X$ as defined in~\eqref{eq:augm_GARCH_pq_1} and~\eqref{eq:augm_GARCH_pq_2} satisfying condition $(Lee)$, $(C_0)$ at $0$ for $r=1$, and both conditions $(C_2^{~'}), (P)$ at $q_X(p)$.
Assume also conditions $(M_r),
 (A)$, and either $(P_{max(1,r/\delta)})$ for $X$ belonging to the group of polynomial GARCH, or $(L_r)$ for the group of exponential GARCH.
Introducing the random vector 
$T_{n,r}(X) = \begin{pmatrix} q_n(p) -q_X(p) \\ \hat{m}(X,n,r) -m(X,r) \end{pmatrix}$, we have the following FCLT: For $t \in [0,1]$, as $n\to\infty$,
\[  \sqrt{n}~t ~ T_{[nt],r}(X) \overset{D_2[0,1]}{\rightarrow}  \textbf{W}_{\Gamma^{(r)}} (t), \]
where $(\textbf{W}_{\Gamma^{(r)}}(t))_{t \in [0,1]}$ is the 2-dimensional Brownian motion with covariance matrix $\Gamma^{(r)} \in \R^{2\times 2}$ defined for any $(s,t) \in [0,1]^2$ by $\Cov(\textbf{W}_{\Gamma^{(r)}}(t),\textbf{W}_{\Gamma^{(r)}}(s)) = \min(s,t) \Gamma^{(r)}$, where
\begin{align*}
\Gamma_{11}^{(r)}&= \Var(W),
\\ \Gamma_{22}^{(r)} &= r^2 \E[ X_0^{r-1} \sgn(X_0)^r]^2 \Var(U) + \Var(V) - 2r \E[ X_0^{r-1} \sgn(X_0)^r] \Cov(U,V),
\\  \Gamma_{12}^{(r)}&= \Gamma_{21}^{(r)} = -r\E[ X_0^{r-1} \sgn(X_0)^r] \Cov(U,W) + \Cov(V,W),
\end{align*}
$(U, V, W)^T$ being the trivariate normal vector (functionals of $X$) with mean zero and covariance given in $(D)$, all series being absolute convergent.
\end{theorem}

Theorem~\ref{th:qn-abs-central-moment_garch} can also be applied to establish a FCLT for $e_n (p) = \widehat{\VaR}_n (\kappa^{-1}(p))$ for $\kappa$ given.
It can be directly extended to a FCLT for a $k$-vector of estimators $\widehat{\VaR}_{n}(p_i), i=1,...,k$. Applying then the continuous mapping theorem yields the case of $\widetilde{\ES}_{n,k} (p)$.

To establish the asymptotics with $\widehat{\ES}_n(p)$, analogously to Proposition~\ref{prop:ES-abs-central-moment} in the iid case, we will need a further dependence condition on the underlying process, namely, strong mixing with a geometric rate (recall Definition~\ref{def:strong_mixing}).
	
To establish the bivariate FCLT for $\widehat{\ES}_n (p)$, we proceed similarly to the case of $\widehat{\VaR}_n(p)$  and introduce, to ease the presentation of the FCLT, a 4-dimensional normal random vector (functionals of $X$), $(U, V, \tilde{W},R)^T$, with mean zero and the following covariance matrix:
\begin{equation*}
 \hspace*{-0.2cm}   (\tilde{D})~\left\{\begin{aligned}
\Var(U) &= \Var(X_0) +2 \sum_{i=1}^{\infty} \Cov(X_i, X_0),
\\ \Var(V) &= \Var(\lvert X_0 \rvert^r) +2 \sum_{i=1}^{\infty}  \Cov(\lvert X_i\rvert^r, \lvert X_0 \rvert^r), 
\\ \Var(\tilde{W}) &= q_X^2 (p) \left( \Var \left( \1_{(X_0 \geq q_X(p))} \right) + 2 \sum_{i=1}^{\infty} \Cov \left( \1_{(X_i \geq q_X(p))}, \1_{(X_0 \geq q_X(p))} \right) \right),
\\ \Var(R) &=  \Var \left( X_0 \1_{(X_0 \geq q_X(p))} \right) + 2 \sum_{i=1}^{\infty} \Cov \left( X_i \1_{(X_i \geq q_X(p))}, X_0 \1_{(X_0 \geq q_X(p))} \right) ,
\\ \Cov(U,V) &= \sum_{i \in \mathbb{Z}} \Cov(\lvert X_i \rvert^r,X_0) = \sum_{i \in \Z} \Cov(\lvert X_0  \rvert^r, X_i  ),
\\ \Cov(U,\tilde{W}) &= q_X(p) \sum_{i \in \mathbb{Z}} \Cov(\1_{(X_i \geq q_X(p))},X_0) = q_X(p) \sum_{i \in \Z} \Cov( \1_{(X_0 \geq q_X(p))},X_i ),
\\ \Cov(V,\tilde{W}) &= q_X(p) \sum_{i \in \mathbb{Z}} \Cov(  \lvert X_0 \rvert^r, \1_{(X_i \geq q_X(p))}) =q_X(p)  \sum_{i \in \Z} \Cov(\lvert X_i \rvert^r, \1_{(X_0 \geq q_X(p))} ),
\\ \Cov(\tilde{W},R) &= q_X(p) \sum_{i \in \Z} \Cov(X_i \1_{(X_i \geq q_X(p))}, \1_{(X_0 \geq q_X(p))} ) = q_X(p) \sum_{i \in \Z} \Cov(X_0 \1_{(X_0 \geq q_X(p))}, \1_{(X_i \geq q_X(p))} ),
\\ \Cov(U,R) &= \sum_{i \in \mathbb{Z}} \Cov(X_i \1_{(X_i \geq q_X(p))},X_0) =  \sum_{i \in \Z} \Cov( X_0 \1_{(X_0 \geq q_X(p))},X_i ),
\\ \Cov(V,R) &=  \sum_{i \in \mathbb{Z}} \Cov(  \lvert X_0 \rvert^r, X_i \1_{(X_i \geq q_X(p))}) =  \sum_{i \in \Z} \Cov(\lvert X_i \rvert^r, X_0 \1_{(X_0 \geq q_X(p))} ).
 \end{aligned}
    \right.
\end{equation*}
Using this 4-dimensional vector, we can now describe the joint asymptotic distribution of $\widehat{\ES}_n (p)$ and $\hat{m}(X,n,r)$.
\begin{proposition} \label{prop:ES-abs-central-moment-garch}
Consider an augmented GARCH($p$, $q$) process $X$ as defined in~\eqref{eq:augm_GARCH_pq_1} and~\eqref{eq:augm_GARCH_pq_2} satisfying the $(Lee)$ condition.
For any integer $r>0$, assume that: $(M_r)$ and $(A)$ hold, $F_X$ is absolutely continuous, $(C_3)$ holds in a neighbourhood of $q_X(p)$, and all the 2nd partial derivatives of the joint distribution of $(X_1,X_{k+1})$, for $k\geq 1$, are bounded in a neighbourhood of $q_X(p)$. 
Assume also either $(P_{max(1, \, r/\delta)})$ for polynomial GARCH, or $(L_r)$ for exponential GARCH and, if $r=1$, $(C_0)$ at the mean $\mu$ and $(M_{r+\delta})$ for some $\delta>0$.

\textbf{If the process is strongly mixing with geometric rate}, introducing the random vector 
$T_{n,r}(X) = \begin{pmatrix} \widehat{ES}_n(p) -\ES(p) \\ \hat{m}(X,n,r) -m(X,r) \end{pmatrix}$, for $r \in \mathbb{Z}$, we have the following FCLT: For $t \in [0,1]$, as $n\to\infty$,
\[  \sqrt{n}~t ~ T_{[nt],r}(X) \overset{D_2[0,1]}{\rightarrow}  \textbf{W}_{\Gamma^{(r)}} (t), \]
where $(\textbf{W}_{\Gamma^{(r)}}(t))_{t \in [0,1]}$ is the 2-dimensional Brownian motion with covariance matrix $\Gamma^{(r)} \in \R^{2\times 2}$ defined for any $(s,t) \in [0,1]^2$ by $\Cov(\textbf{W}_{\Gamma^{(r)}}(t),\textbf{W}_{\Gamma^{(r)}}(s)) = \min(s,t) \Gamma^{(r)}$, where
\begin{align*}
\Gamma_{11}^{(r)}&= \Var(\tilde{W}) + \Var(R) -2 \Cov(\tilde{W},R) ,
\\ \Gamma_{22}^{(r)} &= r^2 \E[ X_0^{r-1} \sgn(X_0)^r]^2 \Var(U) + \Var(V) - 2r \E[ X_0^{r-1} \sgn(X_0)^r] \Cov(U,V),
\\  \Gamma_{12}^{(r)}&= \Gamma_{21}^{(r)} = \Cov(R,V) - \Cov(\tilde{W},V) -r\E[ X_0^{r-1} \sgn(X_0)^r]  \Cov(R,U) 
\\ &\phantom{= \Gamma_{21}^{(r)} =}+ r\E[ X_0^{r-1} \sgn(X_0)^r] \Cov(\tilde{W},U),
\end{align*}
$(U, V, \tilde{W},R)^T$ being the 4-dimensional normal vector (functionals of $X$) with mean zero and covariance given in $(\tilde{D})$, all series being absolute convergent.
\end{proposition}

\begin{remark}
How restrictive is the condition of strong mixing with geometric rate for the augmented GARCH($p$, $q$) processes? 
While we cannot give a general result covering all cases, there exist different results in the literature linking GARCH processes and strong mixing:
Boussama proves in \cite{Boussama98}, Theorem 3.4.2, the strong mixing with geometric rate of a GARCH(p,q) process. 
Carrasco and Chen in \cite{Carrasco02} prove in Proposition~5(i),
that a big class of augmented GARCH(1,1) processes are strongly mixing with geometric rate. 
Therein, in Proposition 12, they also prove strong mixing with geometric rate for the power GARCH($p$,$q$) (PGARCH).
\end{remark}

\begin{remark}
Comparing the conditions in Proposition~\ref{prop:ES-abs-central-moment-garch} with those for $\widehat{\VaR}_n (p)$ in Theorem~\ref{th:qn-abs-central-moment_garch}, we see that we need here the absolute continuity of $F_X$ and the continuity of the second derivative of $f_X$ in a neighbourhood of $q_X(p)$ (instead of $(C_2^{~'})$ and $(P)$ at $q_X(p)$).
Also, for $r=1$, we need $(M_{1+\delta})$ instead of $(M_1)$. These extra conditions are as in the iid case, see Remark~\ref{rmk:iid-ES-asympt}.
But in Proposition~\ref{prop:ES-abs-central-moment-garch}, we also need the process $X$ to be strongly mixing with geometric rate, as well as all second partial derivatives of the joint distribution of $(X_1,X_{k+1})$, for $k\geq 1$, to be bounded (in a neighbourhood of $q_X(p)$). These conditions come from using the Bahadur representation of the ES in \cite{Chen08}.
\end{remark}

\subsection{Proofs (IID models)}
\begin{proof}[Proof of Proposition~\ref{prop:ES-abs-central-moment}]
The proof follows the same ideas as the CLT between the sample quantile and the r-th absolute centred sample moment (Theorem~\ref{th:qn-abs-central-moment}, whose proof can be found in \cite{Brautigam19_iid}).
Only that, instead of using a Bahadur representation for the sample quantile, we use the Bahadur representation for $\widehat{\ES}_{n}(p)$ from~\cite{Chen08}. Since by assumption, $F_X$ is absolutely continuous, $(C_3)$ holds in a neighbourhood of $q_X(p)$, as well as $(M_{1+\delta})$ (or even stronger moment conditions), we can use the ES representation from~\cite{Chen08}:
\begin{equation} \label{eq:Bahadur_ES_iid}
\widehat{ES}_n(p) - ES_P(X) = \frac{1}{(1-p) n}  \sum_{i=1}^n (X_i-q_X(p)) \1_{\left( X_i \geq q_X(p) \right)} - (ES_p (X) - q_X(p)) + o_P (n^{-3/4 + \kappa}),
\end{equation}
for an arbitrary $\kappa>0$.

Accordingly, we know the representation for $\hat{m}(X,n,r)$ (from Proposition~10 in \cite{Brautigam19_iid}): As both, $(C_0)$ at $\mu$ for $r=1$ and $(M_r)$ hold, we have, as $n\rightarrow \infty$,
\begin{equation} \label{eq:mhat_iid_again}
\sqrt{n} \left(\frac{1}{n} \sum_{i=1}^n \lvert X_i - \bar{X}_n \rvert ^r \right) = \sqrt{n} \left(\frac{1}{n} \sum_{i=1}^n \lvert X_i - \mu \rvert^r \right) - r \sqrt{n} (\bar{X}_n - \mu) \E[ (X - \mu)^{r-1} \sgn(X - \mu)^r] + o_P(1).
\end{equation}
Using these two representations, we apply the bivariate CLT. By Slutsky's theorem, we know that we can ignore the remainder terms, which converge in probability to $0$, as they do not change the limiting distribution. The covariance of the asymptotic distribution then simply equals the covariance of the i-th term of \eqref{eq:Bahadur_ES_iid} and \eqref{eq:mhat_iid_again}, respectively,
\begin{align*}
\Cov\left (\frac{1}{1-p} \left( X -q_X(p) \right) \1_{(X \geq q_X(p))}, \lvert X - \mu \rvert^r - r \E[(X- \mu)^{r-1} \sgn(X- \mu)^r] (X - \mu) \right),
\end{align*}
which can be simplified in some cases (e.g. location-scale distributions).

As a last step, we need to identify the variances in the asymptotic distribution. The variance for $\hat{\mu}(X,n,r)$ follows from Proposition~10 in \cite{Brautigam19_iid}, the one for $\widehat{\ES}_n (p)$ from \cite{Chen08}. They are, respectively,
\begin{align*}
& \Var\left(\lvert X - \mu \rvert^r -r (X-\mu) \E[(X-\mu)^{r-1} \sgn(X-\mu)^r] \right),
\\ &\Var\left(\frac{1}{1-p} (X -q_X(p)) \1_{(X \geq q_X(p))}\right).
\end{align*}
Hence, the result~\eqref{eq:cor_ES_var_or_MAD} holds.
\end{proof}

\subsection{Proofs (augmented GARCH($p$,$q$) models)}
\begin{proof}[Proof of Proposition~\ref{prop:ES-abs-central-moment-garch}]
The proof follows the lines of the corresponding FCLT between the sample quantile and the r-th absolute centred sample moment (Theorem~3 in \cite{Brautigam19_garch}), also keeping the same structure of the proof in four steps.

\vspace{0.5cm}
{\sf Step 1: Bahadur representation of the ES - conditions.}\\
As in the proof of Proposition~\ref{prop:ES-abs-central-moment}, we want to use the Bahadur representation of the ES. It holds under the necessary conditions (i) and (ii) as given in \cite{Chen08}, which are fulfilled by assumption:
\begin{itemize}
\item[(i)] The process $X$ is strongly mixing with geometric rate.
\item[(ii)] The stationarity of the process follows from assumption $(P_{\max(1, \, r/\delta)})$ or $(L_r)$, respectively, with Lemma~1 of \cite{Lee14}. The conditions on continuity and moments imposed by \cite{Chen08} are fulfilled by assumption, namely, the absolute continuity of $F_X$, continuous second derivative of $f_X$ in a neighbourhood of $q_X(p)$, the boundedness in a neighbourhood of $q_X(p)$ of all 2nd partial derivatives of the joint distribution of $(Y_1, Y_{k+1})$ for $k \geq 1$.
\end{itemize}
Thus, we can apply the Bahadur representation of the ES
\begin{equation} \label{eq:bahadur_ES}
\widehat{ES}_n(p) - ES(p) = \frac{1}{(1-p) n} \sum_{i=1}^n (X_i-q_X(p)) \1_{\left( X_i \geq q_X(p) \right)} - (ES(p) - q_X(p)) + o_P (n^{-3/4 + \kappa}),
\end{equation}
for an arbitrary $\kappa>0$.

\vspace{0.5cm}
{\sf Step 2: Representation of the r-th absolute centred sample moment -conditions.}\\
This step is exactly the same as in the proof of Theorem~3 in \cite{Brautigam19_garch}.

\vspace{0.5cm}
{\sf Step 3: Conditions for applying the FCLT}\\
This step follows closely Step 3 in the proof of Theorem~3 in \cite{Brautigam19_garch}, adapted to the ES instead of the VaR. 
Here we are using a four-dimensional version of the FCLT (Lemma~9 in \cite{Brautigam19_garch}, choosing $d=4$) - in contrast to a three-dimensional in \cite{Brautigam19_garch}.

Anticipating the use of this Lemma in Step 4 to establish the FCLT for $U_n (X):=\frac{1}{n} \sum_{j=1}^n u_j$, where
\[ u_j = \begin{pmatrix} X_j \\ \lvert X_j  \rvert^r - m(X,r) \\ q_X(p) \1_{(X_j \geq q_X(p))} - (1-p)q_X(p) \\ X_j \1_{(X_j \geq q_X(p))} - \E[X_j \1_{(X_j \geq q_X(p))}] \end{pmatrix}, \]
we verify that the conditions of Lemma~9 in \cite{Brautigam19_garch} hold (equations $(8)$-$(11)$ in \cite{Brautigam19_garch})
$u_j$ fulfills~$(8)$ as $\E[u_j] =0$ holds by construction, and $\E[ \lvert X_j\rvert^{2r}]< \infty$ is guaranteed since $\lvert X_t \rvert^{r}$ satisfies a CLT (see Step 2), thus also $\E[u_j^2] < \infty$. 
As we assume $(A)$, it follows from Lemma 1 in \cite{Lee14} that $X_j = {f}({\epsilon}_j,{\epsilon}_{j-1},...)$. This latter relation also holds for functionals of $X_j$, i.e. $u_j$, thus $(9)$ holds.

Then, we define a $\Delta$-dependent approximation $u_0^{(\Delta)}$ satisfying $(10)$ and~$(11)$.
Denote, for the ease of notation, $X_{0\Delta} := \E[X_0 \vert \mathcal{F}_{-\Delta}^{+\Delta}]$, and set 
\[ u_0^{(\Delta)} = \begin{pmatrix}
X_{0\Delta} \\ \E[\lvert X_0  \rvert^r \vert \mathcal{F}_{-\Delta}^{+\Delta}] - m(X,r) \\ q_X(p) \1_{(X_{0\Delta} \geq q_X(p))} -(1-p) q_X(p) \\ X_{0\Delta} \1_{(X_{0\Delta} \geq q_X(p))} - \E[X_j \1_{(X_j \geq q_X(p))}] \end{pmatrix} \]
with $\mathcal{F}_s^t = \sigma({\epsilon}_s,...,{\epsilon}_t)$ for $s\leq t$.
Thus, $(10)$ is fulfilled by construction. Let us verify~$(11)$. We can write
\begin{align*}
\sum_{\Delta \geq 1} \| u_0 - u_0^{(\Delta)} \|_2  &\leq   \sum_{\Delta \geq 1} \left( \| X_0- X_{0\Delta} \|_2 + \| \lvert X_0\rvert^r - \E[\lvert X_0  \rvert^r \vert \mathcal{F}_{-\Delta}^{+\Delta}] \|_2 \right.
\\  &\left. + q_X^2(p)) \left\| \1_{(X_0 \geq q_X(p))} -  \1_{(X_{0\Delta} \geq q_X(p))}  \right\|_2 + \left\| X_0 \1_{(X_0 \geq q_X(p))} -  X_{0\Delta} \1_{(X_{0\Delta} \geq q_X(p))}  \right\|_2 \right). \numberthis \label{eq:ineq1}
\end{align*}
Since we have already shown the finiteness for the first three parts of the sum in \eqref{eq:ineq1} (in Step~3 of the proof of Theorem~3 in \cite{Brautigam19_garch}), we only need to consider the fourth sum. 
This follows directly by a small algebraic manipulation. Using first the triangle inequality, then the H{\"o}lder inequality (with $p,q \in[1, \infty]$ such that $\frac{1}{p}+ \frac{1}{q} = 1$), we have

\begin{align*}
\small \left\| X_0 \1_{(X_0 \geq q_X(p))} -  X_{0\Delta} \1_{(X_{0\Delta} \geq q_X(p))}  \right\|_2 &= \left\| X_0 ( \1_{(X_{0} \geq q_X(p))} - \1_{(X_{0\Delta} \geq q_X(p))} ) + \1_{(X_0 \Delta \geq q_X(p))} (X_0  -  X_{0\Delta})  \right\|_2
\\ &\leq  \left\| X_0 ( \1_{(X_{0} \geq q_X(p))} - \1_{(X_{0\Delta} \geq q_X(p))} ) \right\|_2 +  \left\| \1_{(X_0 \Delta \geq q_X(p))} (X_0  -  X_{0\Delta})  \right\|_2
\\ &\leq \left\| X_0 \right\|_{2p} \left\| \1_{(X_{0} \geq q_X(p))} - \1_{(X_{0\Delta} \geq q_X(p))} \right\|_{2q} +    \left\| X_0  -  X_{0\Delta}  \right\|_2.
\end{align*}
Choosing $p= 1+ \delta$, for $\delta$ as in Proposition~\ref{prop:ES-abs-central-moment-garch}, $\left\| X_0 \right\|_{2+2\delta}$ is finite by assumption. 
Further, note that we can write, for any $q$,
\[ \left\| \1_{(X_{0} \geq q_X(p))} - \1_{(X_{0\Delta} \geq q_X(p))} \right\|_{2q} = \left\| \1_{(X_{0} \geq q_X(p))} - \1_{(X_{0\Delta} \geq q_X(p))} \right\|_{2}^{1/q}. \]
Then, recall that we know from Step~3 in the proof of Theorem~3 in \cite{Brautigam19_garch} that $\sum_{\Delta \geq 1} \left\| X_0  -  X_{0\Delta}  \right\|_2 < \infty $ and $ \|  \1_{(X_{0} \geq q_X(p))} - \1_{(X_{0\Delta} \geq q_X(p))} \|_2  = O(e^{-\kappa \Delta})$ for some $\kappa>0$. Thus, $\sum_{\Delta \geq 1} \|  \1_{(X_{0} \geq q_X(p))} - \1_{(X_{0\Delta} \geq q_X(p))} \|_2^{1/q}$ is finite. Hence, we can conclude
\[ \sum_{\Delta \geq 1} \left\| X_0 \1_{(X_0 \geq q_X(p))} -  X_{0\Delta} \1_{(X_{0\Delta} \geq q_X(p))}  \right\|_2 < \infty,\]
which means that~$(11)$ is fulfilled.

\vspace{0.5cm}
{\sf Step 4: Multivariate FCLT}\\
Having checked the conditions for the FCLT of Lemma~9 of \cite{Brautigam19_garch} in Step 3, we can apply a 4-dimensional FCLT for $u_j$
\begin{equation} \label{eq:asympt_MAD_trivariate_normal_ES}
 \sqrt{n} \frac{1}{n} \sum_{j=1}^{[nt]} u_j =   \sqrt{n}~t \begin{pmatrix} \bar{X}_{[nt]}  \\ \frac{1}{[nt]} \sum_{j=1}^{[nt]} \lvert X_j \rvert^r - m(X,r) \\ \frac{q_X(p)}{[nt]} \sum_{j=1}^{[nt]} ( \1_{X_j \geq q_X(p))}- (1-p)) \\ \frac{1}{[nt]} \sum_{j=1}^{[nt]} ( X_j \1_{X_j \geq q_X(p))}- \E[X_j \1_{X_j \geq q_X(p))}]) \end{pmatrix} \overset{D_4[0,1]}{\rightarrow}  \textbf{W}_{\tilde{\Gamma}^{(r)}} (t) \quad \text{~as~} n \rightarrow \infty,
\end{equation}
where $\textbf{W}_{\tilde{\Gamma}^{(r)}}(t), t \in [0,1]$ is the 4-dimensional Brownian motion with covariance matrix ${\tilde{\Gamma}^{(r)}} \in \R^{4\times 4}$, i.e. the components ${\tilde{\Gamma}^{(r)}}_{ij}, 1\leq i,j \leq 4$, satisfy the dependence structure $(\tilde{D})$, with all series being absolutely convergent. 

Recalling the representation of $\hat{m}(X,n,r)$ (Proposition~8 in \cite{Brautigam19_garch}) and the Bahadur representation \eqref{eq:bahadur_ES} of the sample ES (ignoring the remainder terms for the moment), we apply to  \eqref{eq:asympt_MAD_trivariate_normal_ES} the multivariate continuous mapping theorem using the function $f(w,x,y,z) \mapsto (aw+x, b(z-y))$ with 
\\ ${a= -r \E[(X-\mu)^{r-1} \sgn(X-\mu)^r]}$, $b=1/(1-p)$, and obtain
\begin{align*} 
 \sqrt{n}~t  &\begin{pmatrix} a (\bar{X}_{[nt]}) + \frac{1}{[nt]} \sum_{j=1}^{[nt]} \lvert X_j \rvert^r - m(X,r)  \\ \frac{1}{1-p} \left( \frac{1}{[nt]} \sum_{j=1}^{[nt]} \1_{(X_j \geq q_X(p))} (X_j -q_X(p))  - (1-p) ( ES_X(p)- q_X(p)) \right)  \end{pmatrix} \overset{D_2[0,1]}{\rightarrow}  \textbf{W}_{\Gamma^{(r)}} (t). \numberthis \label{eq:asymptot1_MAD_garch}
\end{align*}
As by Slutsky's theorem, a remainder term that converges in probability to $0$, does not change the limiting distribution, we get from \eqref{eq:asymptot1_MAD_garch}, 
\begin{align*}
 \sqrt{n}~t \begin{pmatrix} \hat{m}(X,[nt],r) - m(X,r) \\ \widehat{ES}_{[nt]}(p) - \ES_X(p) \end{pmatrix} \overset{D_2[0,1]}{\rightarrow}  \textbf{W}_{\Gamma^{(r)}} (t),
\end{align*} 
where $\Gamma^{(r)}$ follows from the specifications of $\tilde{\Gamma}^{(r)}$ above and the continuous mapping theorem.
\end{proof}

\section{Proofs of Section~\ref{sec:procycl}} \label{sec:appendix_procycl}

\subsection{Proofs of Subsection~\ref{sec:procycl-iid}}
To prove the theorem in the we first present and prove a lemma. This lemma is set in a more general way than the proposition. Then, we will prove the theorem by arguing why the setting of the lemma applies in this case.

\begin{lemma} \label{lemma:help_iid}
Let $\left(X_1,...,X_n \right)$ be an iid sample of copies from a rv $X$. Assume that, for given functions $f$ and $g$, we have $\Var(f(X)) < \infty$ and $\Var(g(x)) < \infty$, such that the bivariate CLT holds:
\begin{equation} \label{eq:lemma_iid_eq1}
\sqrt{n} \begin{pmatrix} \sum_{j=1}^n (f(X_j) - \E[f(X_j)])/n \\ \sum_{j=1}^n (g(X_j) - \E[g(X_j)])/n \end{pmatrix} \overset{d}{\rightarrow} \mathcal{N}(0, \Gamma),
\end{equation}
for a covariance matrix $\Gamma = (\Gamma_{ij}, 1\leq i,j \leq 2)$.
Define 
\begin{equation} \label{eq:lemma_iid_eq2}
Q_{j} = \begin{cases} 0 & \text{~for~} j \leq \lfloor n/2 \rfloor \\ f(X_j)  & \text{~for~} j > \lfloor n/2 \rfloor \end{cases}, 
Y_{j} = \begin{cases} f(X_j) & \text{~for~} j \leq \lfloor n/2 \rfloor \\ 0  & \text{~for~} j > \lfloor n/2 \rfloor \end{cases}, 
Z_{j} = \begin{cases} g(X_j) & \text{~for~} j \leq \lfloor n/2 \rfloor \\ 0 & \text{~for~} j > \lfloor n/2 \rfloor \end{cases}.
\end{equation}
Denote their sample averages (normalized to mean 0) as 
\begin{equation} \label{eq:lemma_iid_eq3}
\widebar{Q}_n = \sum_{j=1}^n (Q_{j} - \E[Q_j])/n, \bar{Y}_n = \sum_{j=1}^n (Y_{j} - \E[Y_j])/n, \bar{Z}_n = \sum_{j=1}^n (Z_j- \E[Z_j]) /n.
\end{equation}

Then, it holds that 
\begin{equation} \label{eq:lemma_iid_result}
\sqrt{n} \begin{pmatrix} \widebar{Q}_n \\ \bar{Y}_n \\ \bar{Z}_n \end{pmatrix} \overset{d}{\rightarrow} \mathcal{N}(0, \Sigma),
\end{equation}
where the covariance matrix $\Sigma$ satisfies $\Sigma_{ij} = \begin{cases} \Gamma_{11}/2 & \text{~for~} i=j \in \{1,2\},  
\\  \Gamma_{22}/2 & \text{~for~} i=j=3  ,
\\ \Gamma_{12}/2 & \text{~for~} i,j \in \{2,3\} \text{~with~} i\neq j 
\\ 0 & \text{otherwise.}  \end{cases}$
\end{lemma}

\begin{proof}
The proof consists of two steps. As we do not work directly on the $X_j$'s, the first step is to establish univariate CLT's for each of the components of the vector~\eqref{eq:lemma_iid_result} using a CLT (Lindeberg-Feller theorem) for independent but not identically distributed rv's.
Then, in a second step, we argue why we can deduce the trivariate asymptotics directly via Cram{\'e}r-Wold.

\vspace{0.5cm}
{\sf Step 1: Univariate CLT's}\\
The proof for each of the three univariate CLT's is analogous. Thus, we prove it for $Q_j$ and only state the results for the two other cases.

Denote $\E[Q_j] = \mu_j , \Var(Q_j) = \sigma_j^2$ (by assumption,  they are finite) and 
\\ $s_n^2 := \sum_{j=1}^n \sigma_j^2 = n\Var(f(X))/2$.
 
For $\frac{\sum_{j=1}^n (Q_j - \mu_j)}{s_n} \overset{d}{\rightarrow} \mathcal{N}(0,1)$ to hold, we need to verify the so called Lindeberg's condition: For all $\epsilon>0$, we need to show that
\[  \lim_{n \rightarrow \infty} \frac{1}{s_n^2} \sum_{j=1}^n \E[(Q_j - \mu_j)^2 \times \1_{(\lvert Q_j - \mu_j \rvert > \epsilon s_n)}] =0.\]
In our case, this translates to
\begin{align*}
\frac{1}{s_n^2} \sum_{j=1}^n \E[(Q_j - \mu_j)^2 \times & \1_{(\lvert Q_j - \mu_j \rvert > \epsilon s_n)}] 
\\& = \frac{2}{n \Var(f(X))} \sum_{j=1}^{n/2} \E[(f(X_j) - \E[f(X_j)])^2 \times \1_{\left(\lvert f(X_j) -  \E[f(X_j)] \rvert > \epsilon \, n \Var(f(X))/2 \right)}]
\\ &= \frac{1}{\Var(f(X))} \E[(f(X) - \E[f(X)])^2 \times \1_{(\lvert f(X) -  \E[f(X)] \rvert > \epsilon \, n \Var(f(X))/2)}].
\end{align*}
As $\Var(f(X))$ is finite, we know that $\1_{(\lvert f(X) -  \E[f(X)] \rvert/ \Var(f(X)) > \epsilon n /2)} \underset{n \rightarrow \infty}{\rightarrow} 0$ almost surely.
Further, $(f(X) - \E[f(X)])^2 \times \1_{(\lvert f(X) -  \E[f(X)] \rvert > \epsilon \, n \Var(f(X))/2)}$ is dominated by $(f(X) - \E[f(X)])^2$, which by assumption is integrable (as $\Var(f(X)) < \infty$). Thus, by dominated convergence, it follows that
\[ \lim_{n  \rightarrow \infty} \E[(f(X) - \E[f(X)])^2 \times \1_{(\lvert f(X) -  \E[f(X)] \rvert > \epsilon \, n \Var(f(X))/2)}] = 0. \]
Thus,  $\widebar{Q}_n$, defined in~\eqref{eq:lemma_iid_eq3}, satisfies $\sqrt{n} \widebar{Q}_n \overset{d}{\rightarrow} \mathcal{N}(0,\Var(f(X))/2)$, i.e. $\Sigma_{11} = \Var(f(X))/2$.

Similarly, we can conclude that $\sqrt{n} \bar{Y}_n \overset{d}{\rightarrow} \mathcal{N}(0,\Var(f(X))/2)$, i.e. $\Sigma_{22} = \Var(f(X))/2$ and $\sqrt{n} \bar{Z}_n \overset{d}{\rightarrow} \mathcal{N}(0,\Var(g(X))/2)$, i.e. $\Sigma_{33} = \Var(g(X))/2$.

\vspace{0.5cm}
{\sf Step 2: Trivariate CLT}\\
To conclude the trivariate normality, it suffices, using the Cram{\'e}r-Wold Device, to show that all linear combinations of $\widebar{Q}_n, \bar{Y}_n, \bar{Z}_n$ are normally distributed.

For any $a,b,c \in \R$, we establish the CLT for 
$U_j := a \left( Q_j - \E[Q_j] \right) + b \left(Y_j - \E[Y_j] \right) + c (Z_j - \E[Z_j])$, i.e.
\[ \sqrt{n} \sum_{j=1}^n U_j /n  \overset{d}{\rightarrow} \mathcal{N}(0, \sigma^2) ,\]
with $\sigma^2= \displaystyle \lim_{n \rightarrow \infty} s_n^2 /n$ to be determined - analogously to Step~1.
Note that $\E[U_j] = 0$ and
\begin{align*}
&s_n^2 = \sum_{j=1}^n \Var(U_j) 
\\ &= \sum_{j=1}^n \Big( a^2 \Var(Q_j) + b^2 \Var(Y_j) + c^2 \Var(Z_j) + 2ab \Cov(Q_j,Y_j) + 2ac \Cov(Q_j, Z_j) + 2bc \Cov(Y_j, Z_j) \Big)
\\&=  a^2 \frac{n}{2} \Var(f(X)) + b^2 \frac{n}{2} \Var(f(X)) + c^2 \frac{n}{2} \Var(g(X)) + 2 bc \frac{n}{2} \Cov(f(X), g(X)) := n \, d(a,b,c), 
\end{align*} 
where
\begin{equation} \label{eq:temp_eq_here}
d(a,b,c):=  \frac{1}{2} \left( a^2 \Var(f(X)) + b^2 \Var(f(X)) + c^2  \Var(g(X))\right) + bc \Cov(f(X), g(X)) ,
\end{equation}
which is finite by assumption. Lindberg's condition is in this case
\begin{align*}
&\frac{1}{s_n^2} \sum_{j=1}^n \E[U_j^2 \times \1_{(\lvert U_j \rvert > \epsilon \, s_n)}] = \frac{1}{s_n^2} \sum_{j=1}^{n/2} \E[U_j^2 \times \1_{(\lvert U_j \rvert > \epsilon \, s_n)}]
+ \frac{1}{s_n^2} \sum_{j=n/2 +1}^n \E[U_j^2 \times \1_{(\lvert U_j \rvert > \epsilon \, s_n)}]
\\ &= \frac{1}{n \, d(a,b,c)} \frac{n}{2} \E[\left(b \left(f(X)-\E[f(X)]\right) + c\left(g(X)-\E[g(X)]\right) \right)^2 \times \1_{\lvert U_j \rvert > \epsilon \,n \, d(a,b,c))}]
\\ &+ \frac{1}{n \, d(a,b,c)} \frac{n}{2} \E[ a^2 \left(f(X) - \E[f(X)]\right)^2 \times \1_{(\lvert a(f(X) - \E[f(X)]) \rvert > \epsilon \, n \, d(a,b,c))}].
\end{align*}
Again, by dominated convergence we can conclude that this quantity converges to zero and thus establish the CLT, i.e.
\[ \sqrt{n} \bar{U}_n  \overset{d}{\rightarrow} \mathcal{N}(0,\sigma^2), \] 
with $\sigma^2 = d(a,b,c)$.
From the knowledge of the univariate asymptotics of $Q_j, Y_j$ and $Z_j$, respectively, we can deduce from~\eqref{eq:temp_eq_here} that it must hold $\Sigma_{12} = \Sigma_{13} =0$ and $\Sigma_{23} = \Cov(f(X),g(X))/2$ to have the trivariate normality of the asymptotic distribution of $\sqrt{n} \begin{pmatrix} \widebar{Q}_n \\ \bar{Y}_n \\ \bar{Z}_n \end{pmatrix}$ with covariance matrix $\Sigma$.

As $\Gamma_{11} = \Var(f(X)), \Gamma_{12} = \Cov(f(X),g(X)), \Gamma_{22} = \Var(g(X))$, the claims on the relation of $\Sigma$ and $\Gamma$ follow directly.
\end{proof}
Now we can turn to the proof of Theorem~\ref{thm:procycl_iid_formal}.
\begin{proof}
The proof consists of two parts.
In the first part, we show why we can apply Lemma~\ref{lemma:help_iid} to the setting of Theorem~\ref{thm:procycl_iid_formal} to establish trivariate asymptotics.

The second part uses Slutsky's theorem, the Delta method and the continuous mapping theorem to deduce from these trivariate asymptotics the claimed bivariate asymptotics.

\vspace{0.5cm}
{\sf Step 1: Applicability of Lemma~\ref{lemma:help_iid}}\\
Recall that we already know that, for $i=1,...,4$,
\begin{equation} \label{eq:iid-rm-rep}
\zeta_{n,i}(p) = \frac{1}{n} \sum_{j=1}^n ( f_i (X_j) - \E[f_i(X_j)]) + o_P(1/\sqrt{n}),
\end{equation}
with the functions being specified as follows:
\begin{itemize}
\item For $i=1$, $f_1 (X_j) = \frac{\1_{(X_j > q_X(p))}}{f_X(q_X(p))}$ - which follows from the Bahadur representation of the sample quantile, see e.g. \cite{Ghosh71}.
\item For $i=2$, $f_2 (X_j) = \frac{(X_j -q_X(p)) \1_{(X_j > q_X(p))}}{1-p}$ - which follows from the Bahadur representation for $\widehat{ES}_n$, see \eqref{eq:Bahadur_ES_iid}.
\item For $i=3$, $f_3(X_j) = \frac{1}{k} \sum_{l=1}^k \frac{\1_{(X_j > q_X(p_l))}}{f_X(q_X(p_l))}$ - recalling the definition of the corresponding estimator, \eqref{eq:ES_nk}, and using the case $i=1$.
\item For $i=4$, $f_4(X_j) =  \frac{\1_{(X_j > q_X(\kappa^{-1}(p)))}}{f_X(q_X(\kappa^{-1}(p)))}$ - recalling the definition of the corresponding estimator, \eqref{eq:expectile_estim}, and using the case $i=1$.
\end{itemize}
Analogously, we know from Proposition~10 in \cite{Brautigam19_iid} that 
\begin{equation} \label{eq:iid-hatm-repr}
\hat{m}(X,n,r) = \frac{1}{n} \sum_{j=1}^n ( g (X_j) - \E[g(X_j)]) + o_P(1/\sqrt{n}),
\end{equation} 
with $g(X_j) = \lvert X_j -\mu \rvert^r - r \E[(X-\mu)^{r-1} \sgn(X-\mu)^r] (X_j - \mu)$.

Thus, we consider Lemma~\ref{lemma:help_iid} for each choice of $f_i$, $i=1,...,4$, as defined above, combined with $g$.
We can identify, by our construction 
\begin{align}
\zeta_{n/2, \, t+n/2, \, i}(p) -\zeta_i(p) &= \widebar{Q}_n + o_P(1/\sqrt{n}), \label{eq:eq1}
\\ \zeta_{n/2, \,t, \,i}(p) - \zeta_i(p) &= \bar{Y}_n + o_P(1/\sqrt{n}), \label{eq:eq2}
\\ \hat{m}(X,n/2, \, r, \,t) - m(X,r) &= \bar{Z}_n + o_P(1/\sqrt{n}), \label{eq:eq3}
\end{align}
using the definitions in \eqref{eq:lemma_iid_eq2} and \eqref{eq:lemma_iid_eq3}.

By the assumption in Theorem~\ref{thm:procycl_GARCH_formal}, the bivariate CLT between $\zeta_{n,i}$ and $m(X,n,r)$ holds. This implies that $\Var(f_i(X)) < \infty$, and $\Var(g(X)) < \infty$ hold (for each $i=1,...,4$). Thus, the conditions of Lemma~\ref{lemma:help_iid} are fulfilled such that~\eqref{eq:lemma_iid_result} holds.

\vspace{0.5cm}
{\sf Step 2: Concluding the bivariate asymptotics}\\
By Slutsky theorem, we know that adding a rest which converges in probability to $0$, does not change the limiting distribution, thus, from equations~\eqref{eq:eq1}-\eqref{eq:eq3} and \eqref{eq:lemma_iid_result}, it follows that, as $n \rightarrow \infty$,
\begin{equation} \label{eq:rm_iid_result1}
\sqrt{n} \begin{pmatrix} \zeta_{n/2, \, t+n/2, \, i}(p) -\zeta_i(p) \\ \zeta_{n/2, \,t, \,i}(p) - \zeta_i(p) \\ \hat{m}(X,n/2, \, r, \,t) - m(X,r) \end{pmatrix} \overset{d}{\rightarrow} \mathcal{N}(0, \Sigma),
\end{equation}
with the covariance matrix $\Sigma$ being related to $\Gamma$ as described in Lemma~\ref{lemma:help_iid}.
By the multivariate Delta method, 
we can deduce from \eqref{eq:rm_iid_result1} that, as $n \rightarrow \infty$,
\begin{equation} \label{eq:rm_iid_result2}
\sqrt{n} \begin{pmatrix} \log\lvert \zeta_{n/2, \, t+n/2, \, i}(p)\rvert -\log\lvert \zeta_i (p) \rvert \\ \log\lvert \zeta_{n/2, \,t, \,i}(p) \rvert - \log\lvert \zeta_i(p) \rvert \\ \hat{m}(X,n/2, \, r, \,t) - m(X,r) \end{pmatrix} \overset{d}{\rightarrow} \mathcal{N}(0, \tilde{\Sigma}),
\end{equation}
where $\tilde{\Sigma}_{jk} = \begin{cases}
\Sigma_{jk}/\zeta_i^2 (p)  & \text{~for~} j,k \in \{ 1,2 \}, 
\\ \Sigma_{jk} & \text{~for~} j=k=3, 
\\ \Sigma_{jk}/\zeta_i (p)  & \text{~else~} \end{cases}$.
\\ Applying the continuous mapping theorem to \eqref{eq:rm_iid_result2} with the function $f(x,y,z) = (x-y,z)$, we obtain
\begin{equation*}
\sqrt{n} \begin{pmatrix} \log \lvert \zeta_{n/2, \, t+n/2, \, i} (p)\rvert - \log \lvert \zeta_{n/2, \,t, \,i} (p) \rvert   \\ \hat{m}(X,n/2, \, r, \,t) - m(X,r) \end{pmatrix} \overset{d}{\underset{n \rightarrow \infty}{\rightarrow}} \mathcal{N}(0, \hat{\Sigma}),
\end{equation*}
where $\hat{\Sigma}_{jk} = \begin{cases}
\tilde{\Sigma}_{11} + \tilde{\Sigma}_{22} & \text{~for~} j=k=1 ,
\\ \tilde{\Sigma}_{33} & \text{~for~} j=k=2 ,
\\ \tilde{\Sigma}_{13} - \tilde{\Sigma}_{23} & \text{~else.} \end{cases}$.

By tracing back the definitions of $\Sigma$ (see Lemma~\ref{lemma:help_iid}), we see that $\hat{\Sigma}$ equals $\tilde{\Gamma}$ as defined in Theorem~\ref{thm:procycl_iid_formal}, and thus conclude the proof.
\end{proof}

\subsection{Proofs of Subsection~\ref{sec:procycl-garch}}

As in the iid case, we will establish a slightly more general result in a lemma, on which the proof of the theorem will be based. 
We present the more general lemma only as the side result, as we are interested specifically in the pro-cyclicality for augmented GARCH($p$, $q$) processes.  

\begin{lemma} \label{lemma:help_garch}
Consider a univariate, stationary stochastic process $(X_j, j\in \Z)$.
Assume the conditions of Lemma~9 in \cite{Brautigam19_garch} to hold such that, for given real functions $f$ and $g$, the bivariate rv 
\\ $u_j := \begin{pmatrix} f(X_j) - \E[f(X_j)] \\ (g(X_j) - \E[g(X_j)) \end{pmatrix}$ satisfies the FCLT, i.e. 
\begin{equation} \label{eq:lemma_garch_eq1}
\sqrt{n} t \begin{pmatrix} \sum_{j=1}^{[nt]} (f(X_j) - \E[f(X_j)])/[nt] \\ \sum_{j=1}^{[nt]} (g(X_j) - \E[g(X_j))/[nt] \end{pmatrix} \overset{D_2[0,1]}{\rightarrow} \textbf{W}_{\Gamma} (t), \quad \text{~as~} \quad n \rightarrow \infty,
\end{equation}
where $(\textbf{W}_{\Gamma}(t))_{t \in [0,1]}$ is the 2-dimensional Brownian motion with covariance matrix $\Gamma \in \R^{2\times 2}$ defined for any $(s,t) \in [0,1]^2$ by $\Cov(\textbf{W}_{\Gamma}(t),\textbf{W}_{\Gamma}(s)) = \min(s,t) \Gamma$.
Define 
\begin{equation} \label{eq:lemma_garch_eq2}
Q_{j} = \begin{cases} 0 & \text{~for~} j \leq \lfloor n/2 \rfloor \\ f(X_j)  & \text{~for~} j > \lfloor n/2 \rfloor \end{cases}, 
Y_{j} = \begin{cases} f(X_j) & \text{~for~} j \leq \lfloor n/2 \rfloor \\ 0  & \text{~for~} j > \lfloor n/2 \rfloor \end{cases}, 
Z_{j} = \begin{cases} g(X_j) & \text{~for~} j \leq \lfloor n/2 \rfloor \\ 0 & \text{~for~} j > \lfloor n/2 \rfloor \end{cases}.
\end{equation}
Denote their sample averages (normalized to mean 0) as 
\begin{equation} \label{eq:lemma_garch_eq3}
\widebar{Q}_n = \sum_{j=1}^n (Q_{j} - \E[Q_j])/n, \quad \bar{Y}_n = \sum_{j=1}^n (Y_{j} - \E[Y_j])/n,  \quad \bar{Z}_n = \sum_{j=1}^n (Z_j- \E[Z_j]) /n.
\end{equation}
Then, {\bf if the process $X_j$ is strongly mixing with geometric rate and there exists a $\delta>0$ s.t.}
\begin{equation} \label{eq:lemma_garch_assumption}
 \E[\lvert Q_j - \E[Q_j] \rvert^{2+2\delta}] < \infty, \quad \E[\lvert Y_j - \E[Y_j] \rvert^{2+2\delta}] < \infty, \quad \E[\lvert Z_j -\E[Z_j] \rvert^{2+2\delta}] < \infty~,\forall j, 
\end{equation}
it holds that 
\begin{equation} \label{eq:lemma_garch_result}
\sqrt{n} \begin{pmatrix} \widebar{Q}_n \\ \bar{Y}_n \\ \bar{Z}_n \end{pmatrix} \overset{d}{\rightarrow} \mathcal{N}(0, \Sigma),
\end{equation}
where the covariance matrix $\Sigma$ satisfies $\Sigma_{ij} = \begin{cases} \Gamma_{11}/2 & \text{~for~} i=j \in \{1,2\},  
\\  \Gamma_{22}/2 & \text{~for~} i=j=3  ,
\\ \Gamma_{12}/2 & \text{~for~} i,j \in \{2,3\} \text{~with~} i\neq j ,
\\ 0 & \text{otherwise.}  \end{cases}$
\end{lemma}

\begin{proof}
The idea of the proof is the same as in the iid case (as expected from the choice made to prove the iid case). It consists of two steps.
First, we need to establish univariate CLT's for each of the components of the vector in~\eqref{eq:lemma_garch_result}, using a CLT for non-stationary strongly mixing sequences. 
Secondly, we argue why we can deduce the trivariate asymptotics directly via Cram{\'e}r-Wold.
To do so, we need to show that the covariances between estimators over disjoint samples vanish asymptotically. For this, we will use covariance bounds for strongly mixing processes.

\vspace{0.5cm}
{\sf Step 1: Univariate CLT's}\\
To establish the univariate CLT's, we use a CLT for non-stationary sequences by~\cite{Politis97}, \cite{Ekstrom14}, which we simplify to our purposes, as follows:

{\it Consider a stochastic process, denoted by $(W_j, j \in \Z$), which is strongly mixing with coefficient $\alpha(k)$. Denote $\bar{W}_n = \frac{1}{n} \sum_{j=1}^n W_j$ and $\sigma_{n}^2 = \Var(\sqrt{n} \bar{W}_{n})$.
If the following three conditions hold,
\begin{align}
&\E[\lvert W_j - E[W_j] \rvert^{2+2\delta}] \leq c, \quad  \forall j \label{eq:cond1_garch}
\\ &\sigma^2 := \lim_n \sigma_n \in (0, \infty)\label{eq:cond2_garch}
\\ &\sum_{k=0}^{\infty} (k+1)^2 \alpha(k)^{\delta/(4+\delta)} \leq d, \text{~for a finite constant~} d \text{~ independent of~} k, \label{eq:cond3_garch}
\end{align}
\vspace{-2ex}
then $\sqrt{n} (\bar{W}_n - E[\bar{W}_n])   \overset{d}{\rightarrow} \mathcal{N}(0, \sigma^2)$ as $n \rightarrow \infty$.}

Note that a stronger condition than~\eqref{eq:cond2_garch}, is introduced in \cite{Politis97}, namely
\begin{equation} \label{eq:politis_cond}
\forall (d_n) \text{~s.t.~} d_n \rightarrow \infty:  \sup_t \lvert  \Var(\sqrt{d_n} \frac{1}{d_n} \sum_{j=t}^{t+d_n-1} W_j) - \sigma^2 \rvert \rightarrow 0, \text{ ~as~} n \rightarrow \infty,
\end{equation}
under which the authors conclude that $\frac{1}{\sqrt{d_n}} \sum_{i=1}^{d_n} X_i \overset{d}{\rightarrow} \mathcal{N}(0, \sigma^2)$ holds (with $\displaystyle  \sigma^2 := \lim_{n \rightarrow \infty} \sigma_n$) for {\it any} sequence $d_n \leq n$ such that $d_n \rightarrow \infty$ as $n \rightarrow \infty$.
To ensure this, \eqref{eq:politis_cond} is reasonable, i.e. the CLT should hold for any $d_n$ with always the same variance $\sigma^2$.
In our case, we only need the CLT to hold for $d_n=n$ (and we do not care what would happen for other choices of $d_n$).
This is why we consider \cite{Ekstrom14}, who shows that \eqref{eq:politis_cond} is actually superfluous, but at the price of accepting potentially degenerate limiting distributions.
As a compromise between the two, we demand \eqref{eq:cond2_garch}, which ensures that we do not have a degenerate limiting distribution for the case $d_n = n$.
%

The proof for each of the three univariate CLT's is analogous. Thus, we prove it for $Q_j$ and only state the results for the two other cases.

Let us verify the conditions \eqref{eq:cond1_garch} to~\eqref{eq:cond3_garch} so that we can apply the CLT. 
First, we note that \eqref{eq:cond1_garch} corresponds, in our case, to our assumption~\eqref{eq:lemma_garch_assumption}, hence is satisfied. 
Direct computations lead to~\eqref{eq:cond2_garch}:
\begin{align*}
\sigma_Q^2 &= \lim_n \Var( \sqrt{n}\widebar{Q}_n ) = \lim_n \frac{1}{n} ( \sum_{j=n/2+1}^{n} \Var(Q_j) + 2 \sum_{n/2 + 1 \leq i < j \leq n} \Cov(Q_i, Q_j) ) 
\\ &= \Var(f(X_0))/2  + \lim_n \frac{2}{n} \sum_{i=1}^{n/2-1} (n/2 - i) \Cov(f(X_0), f(X_i)) 
\\ &=  \Var(f(X_0))/2 + \sum_{i=1}^{\infty} \Cov(f(X_0), f(X_i)),
\end{align*}
which is non-degenerate by \eqref{eq:lemma_garch_eq1}.

As $Q_j$ is a functional of $X_j$, we can bound from above the mixing coefficient of $Q_j$, denoted by $\alpha_Q (k)$, by the one of $X_j$, i.e.
$\alpha_Q (k) \leq \alpha (k)$.
As we know that $X_j$ is strongly mixing with geometric rate, we have that $\alpha_Q (k) \leq C \lambda^k$ for some constants $C>0$ and $\lambda \in (0,1)$, which implies:
\begin{align*}
\sum_{k=0}^{\infty} (k+1)^2 \alpha_Q(k)^{\delta/(4+\delta)} &\leq \sum_{k=0}^{\infty} (k+1)^2 (C \lambda^k) ^{\delta/(4+\delta)} = C^{\delta/(4+\delta)} \sum_{k=1}^{\infty} k^2  \lambda^{(k-1)\delta/(4+\delta)}.
\end{align*}
We perform a ratio test to confirm the convergence of this series
\[ L = \lim_{k \rightarrow \infty} \left\lvert \frac{(k+1)^2  \lambda^{k \delta/(4+\delta)}}{k^2  \lambda^{(k-1) \delta/(4+\delta)}} \right\rvert = \lim_{k \rightarrow \infty} \left\lvert (1 + \frac{2}{k} + \frac{1}{k^2}) \lambda^{\delta/(4+\delta)} \right\rvert = \lambda^{\delta/(4+\delta)} < 1. \]
Thus, the series is convergent, from which we deduce~\eqref{eq:cond3_garch}. We conclude to the CLT, as $n \rightarrow \infty$
\[ \sqrt{n} (\widebar{Q}_n - E[\widebar{Q}_n])   \overset{d}{\rightarrow} \mathcal{N}(0, \sigma_Q^2). \]

In the same manner, we obtain, as $n \rightarrow \infty$,
\begin{align*}
\sqrt{n} (\bar{Y}_n - E[\bar{Y}_n])   &\overset{d}{\rightarrow} \mathcal{N}(0, \sigma_Y^2) \quad \text{~ and ~} \quad
\sqrt{n} (\bar{Z}_n - E[\bar{Z}_n]) \overset{d}{\rightarrow} \mathcal{N}(0, \sigma_Z^2), \text{~as~} n \rightarrow \infty,
\end{align*} 
where 
\begin{align*}
\sigma_Q^2 &= \sigma_Y^2 \quad \text{~and~} \quad \sigma_Z^2=  \Var(g(X_0))/2 + \sum_{i=1}^{\infty} \Cov(g(X_0), g(X_i)).
\end{align*}

{\sf Step 2: Trivariate CLT}\\
By the Cram{\'e}r-Wold Device, it suffices to show that all linear combinations of the components of $(\widebar{Q}_n, \bar{Y}_n ,\bar{Z}_n)^T$ are normally distributed, to conclude their trivariate normality.

For any $a,b,c \in \R$, we establish the CLT for 
\[ U_j := a \left( Q_j - \E[Q_j] \right) + b \left(Y_j - \E[Y_j] \right) + c (Z_j - \E[Z_j]), \]
i.e.
\[ \sqrt{n} \sum_{j=1}^n U_j /n  \overset{d}{\rightarrow} \mathcal{N}(0, \sigma^2), \text{~as~} n \rightarrow \infty,\]
with $\sigma^2$ to be determined - in a similar way as in Step~1.
Note that, by construction, $\E[U_j] = 0$.
We need to verify the strong mixing of $U_j$ and the three conditions \eqref{eq:cond1_garch} to~\eqref{eq:cond3_garch}.
By the Minkowski inequality, we have that 
\[ \E[ \lvert U_j  \rvert^{2+2\delta}] = \| U_j \|_{2+2\delta}^{2+ 2\delta} \leq (a \| Q_j - \E[Q_j]\|_{2+2\delta} + b \| Y_j - \E[Y_j]\|_{2+2\delta}+ c\| Z_j - \E[Z_j]\|_{2+2\delta})^{2+2\delta}. \]
Thus, \eqref{eq:cond1_garch} is fulfilled by assumption, \eqref{eq:lemma_garch_assumption}.

By construction, each $U_j$ is a functional of $X_j$ (which is strongly mixing with geometric rate, by assumption). We can bound from above the mixing coefficient of $U_j$, denoted by $\alpha_U (k)$, by the one of $X_j$, i.e.
$\alpha_U (k) \leq \alpha (k)$. Therefore, \eqref{eq:lemma_garch_eq3} holds by the same argumentation as in the univariate case.

So, we are left with computing $\displaystyle \sigma^2 = \lim_{n \rightarrow \infty} \sigma_n^2$. We write it as:
\begin{align*}
\sigma_n^2 &= \Var(\sqrt{n} \sum_{j=1}^n U_j/n) = \frac{1}{n}  \Var(\sqrt{n} \sum_{j=1}^n (a Q_j + bY_j +cZ_j)/n) \notag
\\ &= a^2 \Var( \sqrt{n} \sum_{j=1}^n Q_j/n) + b^2 \Var( \sqrt{n} \sum_{j=1}^n Y_j/n) + c^2  \Var( \sqrt{n} \sum_{j=1}^n Z_j/n) 
\\ &+ 2ab \Cov(\sqrt{n} \sum_{j=1}^n Q_j/n, \sqrt{n} \sum_{i=1}^n Y_i/n) + 2ac  \Cov(\sqrt{n} \sum_{j=1}^n Q_j/n, \sqrt{n} \sum_{i=1}^n Z_i/n)  
\\ & + 2 bc  \Cov(\sqrt{n} \sum_{j=1}^n Y_j/n, \sqrt{n} \sum_{i=1}^n Z_i/n). \numberthis \label{eq:lemma-garch-covariances}
\end{align*}
As this expression for $\sigma_n^2$ is quite long and some computations will be involved, we split the computation into different parts.
First, note that the respective variances in \eqref{eq:lemma-garch-covariances} are known from the univariate asymptotics: 
\begin{equation} \label{eq:lemma-garch-variances-lim}
 \lim_{n \rightarrow \infty} \Var( \sqrt{n} \sum_{j=1}^n Q_j/n) = \sigma_Q^2, \quad
\lim_{n \rightarrow \infty} \Var( \sqrt{n} \sum_{j=1}^n Y_j/n) = \sigma_Y^2, \quad
\lim_{n \rightarrow \infty} \Var( \sqrt{n} \sum_{j=1}^n Z_j/n) = \sigma_Z^2.
\end{equation}
Thus, we are left with the covariances which we assess one after the other. 

{$\bullet$ Computation of the first covariance of \eqref{eq:lemma-garch-covariances}}
\begin{align*}
\Cov(\sqrt{n} \sum_{j=1}^n Q_j/n, &\sqrt{n} \sum_{i=1}^n Y_i/n) 
\\  &=  \frac{1}{n} \sum_{j=n/2 +1}^n \sum_{i=1}^{n/2} \Cov(f(X_j), f(X_i))
\\ &=  \frac{1}{n} \sum_{j=n/2 +1}^n \sum_{i=1}^{n/2} \Cov(f(X_{j-i}), f(X_0))
\\ &=  \frac{1}{n} \left( \sum_{k=1}^{n/2}  k \Cov(f(X_{k}), f(X_0)) + \sum_{k=n/2+1}^{n -1} (n-k) \Cov(f(X_{k}), f(X_0)) \right)
\\ &=  \frac{1}{n} \left( \sum_{k=1}^{n/2}  k \Cov(f(X_{k}), f(X_0)) + \sum_{k=1}^{n/2 -1} (\frac{n}{2} -k) \Cov(f(X_{k+n/2}), f(X_0)) \right), \numberthis \label{eq:lemma-garch_covsum1}
\end{align*}
where we used the stationarity of the underlying process $X$.

To bound the two sums in \eqref{eq:lemma-garch_covsum1}, we use covariance bounds provided in \cite{Roussas87}, Theorem~7.3. We recall them here, for convenience, for a process $(X_j, j \in \Z)$:
{ \it
\begin{itemize}
\item[-] If $f(X_k)$ is $\mathcal{F}_{l+k}^{\infty}$ measurable and $f(X_0)$ is $\mathcal{F}_{-\infty}^l$ measurable (for a chosen integer $l$ and $k>0$),
\item[-] if $\E[ \lvert f(X_0) \rvert^p] < \infty$ and $\E[\lvert f(X_k) \rvert^{rq}] < \infty$ for some $p,q>1$ s.t. ${\frac{1}{p}+\frac{1}{q}<1}$,
\item[-] if $X_j, f  \in \Z$, is strongly mixing, with mixing coefficient $\alpha(k)$ 
\end{itemize}
then we have $\displaystyle \lvert \Cov(f(X_0), f(X_k)) \rvert \leq 10  ~\alpha(k)^{1-\frac{1}{p}-\frac{1}{q}} \| f(X_0) \|_p \| f(X_k )\|_q$ .}

Choosing $q=2$ and $p=2+2\delta$ (as, by~\eqref{eq:lemma_garch_assumption}, those moments will exist), and $l=0$, we can write the inequality above as
\[ \lvert \Cov(f(X_0), f(X_k)) \rvert \leq M ~\alpha(k)^{1-\frac{1}{p}-\frac{1}{q}}, \]
where $M:= 10 \| f(X_0) \|_p \| f(X_k )\|_q $.

Recall, as the process is strong mixing with geometric rate, that there exist constants $C>0$ and $\lambda \in (0,1)$ s.t. $\alpha(k) \leq C \lambda^k$.
We use this geometric rate and the covariance bound to show the finiteness of the first covariance sum of~\eqref{eq:lemma-garch_covsum1}:
\begin{align*}
\sum_{k=1}^{n/2}  k \Cov(f(X_{k}), f(X_0)) \leq \sum_{k=1}^{n/2}  k  \lvert \Cov(f(X_{k}), f(X_0)) \rvert
&\leq \sum_{k=1}^{n/2}  k  M \alpha(k)^{1-\frac{1}{p}-\frac{1}{q}}
\\ & \leq M C^{1-\frac{1}{p}-\frac{1}{q}} \sum_{k=1}^{n/2}  k \lambda^{k (1-\frac{1}{p}-\frac{1}{q})}.
\end{align*}
Using once again the ratio test for the finiteness of the latter series (as $n \rightarrow \infty$)
\[ L = \lim_{k \rightarrow \infty} \left\lvert \frac{ (k+1) \lambda^{ (k+1) (1-\frac{1}{p}-\frac{1}{q})}}{ k \lambda^{k (1-\frac{1}{p}-\frac{1}{q})}} \right\rvert = \lim_{k \rightarrow \infty} (1+1/k) \lambda^{ (1-\frac{1}{p}-\frac{1}{q})} = \lambda^{ (1-\frac{1}{p}-\frac{1}{q})} <1, \]
we deduce that
\begin{equation} \label{eq:lemma-garch_covsum1-res1}
\lim_n \frac{1}{n}  \sum_{k=1}^{n/2}  k \Cov(f(X_{k}), f(X_0)) = 0.
\end{equation}
Now we need to look at the second sum of~\eqref{eq:lemma-garch_covsum1}. We proceed in the same way using the strong mixing rate as well as the covariance bounds:
\begin{align*}
\frac{1}{n} \sum_{k=1}^{n/2-1} (\frac{n}{2}-k) \Cov ( f(X_{k+n/2}),f(X_0)) &\leq \frac{1}{n}\sum_{k=1}^{n/2-1} (\frac{n}{2}-k)  \lvert \Cov ( f(X_{k+n/2}),f(X_0)) \rvert
\\  &\leq \frac{1}{n} \sum_{k=1}^{n/2-1} (\frac{n}{2}-k) M \alpha(k+n/2)^{1-\frac{1}{p}-\frac{1}{q}}
\\ &\leq \frac{1}{n} \sum_{k=1}^{n/2-1} (\frac{n}{2}-k) M (C\lambda^{k+n/2})^{1-\frac{1}{p}-\frac{1}{q}}. \numberthis \label{eq:intopato}
\end{align*}
For the ease of notation, define $\tilde{\lambda} = \lambda^{1-\frac{1}{p}-\frac{1}{q}}$ and $\tilde{M} = M C^{1-\frac{1}{p}-\frac{1}{q}}$, such that we have from~\eqref{eq:intopato}
\begin{align*}
\frac{1}{n} \sum_{k=1}^{n/2-1} (\frac{n}{2}-k) \Cov ( f(X_{k+n/2}),f(X_0))&\leq \tilde{M} \tilde{\lambda}^{n/2}  \sum_{k=1}^{n/2-1} (\frac{1}{2}-\frac{k}{n})  \tilde{\lambda}^k 
\\ & \leq \tilde{M} \tilde{\lambda}^{n/2}  \sum_{k=1}^{n/2-1} (\frac{1}{2}-\frac{k}{n})  = \tilde{M} \tilde{\lambda}^{n/2} \frac{n-2}{8},
\end{align*}
%
which tends to $0$, as $n \to \infty$, as $\tilde{\lambda}<1$. 
Thus, we can conclude that
\begin{equation} \label{eq:lemma-garch_covsum1-res2}
\lim_{n \rightarrow \infty} \frac{1}{n} \sum_{k=1}^{n/2-1} (\frac{n}{2}-k) \Cov ( f(X_{k+n}),f(X_0)) = 0.
\end{equation}
Combining~\eqref{eq:lemma-garch_covsum1} with~\eqref{eq:lemma-garch_covsum1-res1} and~\eqref{eq:lemma-garch_covsum1-res2}, we conclude for the first covariance sum of \eqref{eq:lemma-garch-covariances} that:
\begin{equation} \label{eq:lemma-garch_covsum1-finalresult}
\lim_n \Cov(\sqrt{n} \sum_{j=1}^n Q_j /n, \sqrt{n} \sum_{i=1}^n Y_i /n) =0.
\end{equation} 
\vspace{0.5cm}
{$\bullet$ Computation of the second covariance of \eqref{eq:lemma-garch-covariances}}\\
The computation of the limit of the second covariance of  \eqref{eq:lemma-garch-covariances} is analogous to the first one, simply replacing $Y_i$ by $Z_i$ and thus $f(X_i)$ by $g(X_i)$. I.e. from \eqref{eq:lemma-garch_covsum1} we deduce that
\begin{align*}
 \Cov(\sqrt{n} \sum_{j=1}^n Q_j/n, \, \sqrt{n} \sum_{i=1}^n Z_i/n)  &= \frac{1}{n} \sum_{j=1}^n \sum_{i=1}^n \Cov(Q_j, Z_i) = \cdots
\\ & =  \frac{1}{n} \left( \sum_{k=1}^{n/2}  k \Cov(f(X_{k}), g(X_0)) + \sum_{k=1}^{n/2 -1} (\frac{n}{2}-k) \Cov(f(X_{k+n/2}), g(X_0)) \right). \numberthis \label{eq:lemma-garch_covsum2}
\end{align*}
The covariance bounds are again applicable. Choosing $p=2$ and $ q=2+2\delta$, those moments exist by~\eqref{eq:lemma_garch_assumption}. 
Thus, we obtain analogous results to \eqref{eq:lemma-garch_covsum1-res1} and~\eqref{eq:lemma-garch_covsum1-res2} and can conclude, as for the first covariance of~\eqref{eq:lemma-garch-covariances}, that
\begin{equation} \label{eq:lemma-garch_covsum2-finalresult}
\lim_{n \rightarrow \infty} \Cov(\sqrt{n} \sum_{j=1}^n Q_j/n, \sqrt{n} \sum_{i=1}^n Z_i/n) =0. 
\end{equation} 

\vspace{0.5cm}
{$\bullet$ Computation of the third covariance of \eqref{eq:lemma-garch-covariances}}

We are left with
\begin{align*}
\Cov(\sqrt{n} \sum_{j=1}^n Y_j/n, \sqrt{n} \sum_{i=1}^n Z_i/n) 
 &=  \frac{1}{n} \sum_{j=1}^{n/2} \sum_{i=1}^{n/2} \Cov(f(X_j), f(X_i))
\\ &=  \frac{1}{n} \left( \frac{n}{2} \Cov(f(X_0), f(X_0)) + 2 \sum_{i=1}^{n/2-1} (\frac{n}{2}-i) \Cov(f(X_{i}), f(X_0)) \right). \numberthis \label{eq:lemma-garch_covsum3}
\end{align*}
Thus, we have
\begin{equation} \label{eq:lemma-garch_covsum3-finalresult}
\lim_{n \rightarrow \infty} \Cov(\sqrt{n} \sum_{j=1}^n Y_j/n, \sqrt{n} \sum_{i=1}^n Z_i/n) = \Var(f(X_0))/2 + \sum_{i=1}^{\infty}  \Cov(f(X_{i}), f(X_0)).
\end{equation}

Therefore, we can finally compute $\sigma_n^2$. 
We get, recalling the expressions for the variances in~\eqref{eq:lemma-garch-variances-lim} and for the covariances in~\eqref{eq:lemma-garch_covsum1-finalresult}, \eqref{eq:lemma-garch_covsum2-finalresult} and~\eqref{eq:lemma-garch_covsum3-finalresult}, that
\begin{align*}
\sigma_n^2 &= \Var(\sqrt{n} \sum_{j=1}^n U_j/n) = \frac{1}{n}  \Var(\sqrt{n} \sum_{j=1}^n (a Q_j + bY_j +cZ_j)/n)
\\ &= a^2 \Var( \sqrt{n} \sum_{j=1}^n Q_j/n) + b^2 \Var( \sqrt{n} \sum_{j=1}^n Y_j/n) + c^2  \Var( \sqrt{n} \sum_{j=1}^n Z_j/n) 
\\ &+ 2ab \Cov(\sqrt{n} \sum_{j=1}^n Q_j/n, \sqrt{n} \sum_{i=1}^n Y_i/n) + 2ac  \Cov(\sqrt{n} \sum_{j=1}^n Q_j/n, \sqrt{n} \sum_{i=1}^n Z_i/n) 
\\ &+ 2 bc  \Cov(\sqrt{n} \sum_{j=1}^n Y_j/n, \sqrt{n} \sum_{i=1}^n Z_i/n).
\end{align*}
Hence, we have in the limit
\begin{equation} \label{eq:deducers}
\lim_{n \rightarrow \infty} \sigma_n^2 = a^2 \sigma_Q^2 + b^2 \sigma_Y^2 + c^2 \sigma_Z^2 + 2bc \left( \Var(f(X_0))/2 + \sum_{i=1}^{\infty}  \Cov(f(X_{i}), f(X_0)) \right). 
\end{equation}

Recalling the univariate asymptotics of $\widebar{Q}_n, \bar{Y}_n$ and $\bar{Z}_n$, respectively, $\Sigma_{11} = \sigma_Q^2, \quad \Sigma_{22} = \sigma_Y^2, \quad \Sigma_{33} = \sigma_Z^2$, we can deduce from \eqref{eq:deducers} that it must hold $\Sigma_{12} = \Sigma_{13} =0$ and $\Sigma_{23} = \Var(f(X_0))/2 + \sum_{i=1}^{\infty}  \Cov(f(X_{i}), f(X_0))$ to have the trivariate normality of the asymptotic distribution of $\sqrt{n} \begin{pmatrix} \widebar{Q}_n \\ \bar{Y}_n \\ \bar{Z}_n \end{pmatrix}$ with covariance matrix $\Sigma$.

The claims on the relation of $\Sigma$ and $\Gamma$ follow directly by comparing.
\end{proof}

%
After having proved Lemma~\ref{lemma:help_garch}, which was the main work, we can proceed with the proof of Theorem~\ref{thm:procycl_GARCH_formal}.
\begin{proof}
The proof is structurally the same as in the iid case (on purpose, that is why whe chose to proceed this way for the iid case), we only have to update the references to the corresponding ones for augmented GARCH($p$, $q$) processes. Still, we present the proof briefly. It consists of two parts.
In the first part, we show that we can apply Lemma~\ref{lemma:help_garch} to establish trivariate asymptotics.
The second part uses Slutsky's theorem, the Delta method and the continuous mapping theorem to deduce from the trivariate asymptotics the claimed bivariate asymptotics.

\vspace{0.5cm}
{\sf Step 1: Applicability of Lemma~\ref{lemma:help_garch}}\\
Recall that we already know that, for $i=1,...,4$, \begin{equation} \label{eq:garch-rm-rep}
\zeta_{n,i}(p) = \sum_{j=1}^n ( f_i (X_j) - \E[f_i(X_j)]) /n+ o_P(1/\sqrt{n}),
\end{equation}
with the functions specified as follows:
\begin{itemize}
\item For $i=1$, $f_1 (X_j) = \frac{\1_{(X_j > q_X(p))}}{f_X(q_X(p))}$ - which follows from the Bahadur representation of the sample quantile, see e.g.~\cite{Wendler11}.
\item For $i=2$, $f_2 (X_j) = \frac{(X_j -q_X(p)) \1_{(X_j > q_X(p))}}{1-p}$ - which follows from the Bahadur representation for $\widehat{ES}_n$, see \eqref{eq:bahadur_ES}.
\item For $i=3$, $f_3(X_j) = \frac{1}{k} \sum_{l=1}^k \frac{\1_{(X_j > q_X(p_l))}}{f_X(q_X(p_l))}$ - recalling the definition of the corresponding estimator, \eqref{eq:ES_nk}, and using the case $i=1$.
\item For $i=4$, $f_4(X_j) =  \frac{\1_{(X_j > q_X(\kappa^{-1}(p)))}}{f_X(q_X(\kappa^{-1}(p)))}$ - recalling the definition of the corresponding estimator, \eqref{eq:expectile_estim}, and using the case $i=1$.
\end{itemize}
Analogously, we know from Proposition~8 in \cite{Brautigam19_garch} that 
\begin{equation} \label{eq:garch-hatm-repr}
\hat{m}(X,n,r) = \sum_{j=1}^n ( g (X_j) - \E[g(X_j)]) /n+ o_P(1/\sqrt{n}),
\end{equation} 
i.e. $g(X_j) = \lvert X_j -\mu \rvert^r - r \E[(X-\mu)^{r-1} \sgn(X-\mu)^r] (X_j - \mu)$.

We know that the representations~\eqref{eq:garch-rm-rep} and~\eqref{eq:garch-hatm-repr} hold as, by assumption in Theorem~\ref{thm:procycl_GARCH_formal}, the conditions for the bivariate asymptotics between $\zeta_{n,i}$ and $m(X,n,r)$ are fulfilled. 

Then, we consider Lemma~\ref{lemma:help_garch} for each choice of $f_i$, $i=1,...,4$, as defined above combined with $g$.
We can identify, by our construction, 
\begin{align}
\zeta_{n/2, \, t+n/2, \, i}(p) -\zeta_i(p) &= \widebar{Q}_n + o_P(1/\sqrt{n}), \label{eq:eq1_garch}
\\ \zeta_{n/2, \,t, \,i}(p) - \zeta_i(p) &= \bar{Y}_n + o_P(1/\sqrt{n}), \label{eq:eq2_garch}
\\ \hat{m}(X,n/2, \, r, \,t) - m(X,r) &= \bar{Z}_n + o_P(1/\sqrt{n}), \label{eq:eq3_garch}
\end{align}
using the definitions~\eqref{eq:lemma_garch_eq2} and \eqref{eq:lemma_garch_eq3}.
Again, by assumption in Theorem~\ref{thm:procycl_GARCH_formal}, the bivariate CLT, i.e. \eqref{eq:lemma_garch_eq1}, between $\zeta_{n,i}$ and $m(X,n,r)$ holds. 
As the strong mixing and the moment condition,~\eqref{eq:lemma_garch_assumption}, hold by assumption too, by Lemma~\ref{lemma:help_garch}, the claimed trivariate asymptotics \eqref{eq:lemma_garch_result} hold.

\vspace{0.5cm}
{\sf Step 2: Concluding the bivariate asymptotics}\\
This is exactly the same as the Step~2 in the proof of Theorem~\ref{thm:procycl_iid_formal}, only replacing the use of Lemma~\ref{lemma:help_iid} by Lemma~\ref{lemma:help_garch}, and the covariance matrix from the iid case with the one from the GARCH case.
\end{proof}

\section{Explicit Formulas Corresponding to Examples in Section~\ref{ssec:exp-iid}} \label{appendix:formulas}

As in the plots of Section~\ref{ssec:exp-iid} we consider the asymptotic correlation between either the sample variance or sample MAD as r-th central absolute sample moment with one of the three risk measures $\widehat{\VaR}_{n,t}(p)$, $\widehat{\ES}_{n,t}(p)$ and $e_{n,t}(p)$.

In Table~\ref{tbl-cor-riskmeasures_norm} we present the expressions for an underlying Gaussian distribution and then in Table~\ref{tbl-cor-riskmeasures_stud} for a Student distribution with $\nu$ degrees of freedom.

To show how we obtain the expressions in Tables~\ref{tbl-cor-riskmeasures_norm} and~\ref{tbl-cor-riskmeasures_stud}, we only need to focus on the quantities with the sample ES.

Indeed, for the correlations including the sample VaR, i.e. with the sample variance or the sample MAD there is nothing to do as they are simply the asymptotic correlation of the sample quantile with the sample variance or the sample MAD, respectively - which were already computed in \cite{Brautigam18_iid_WP}.

The same remarks hold for the expectile estimator, as it is the sample quantile at level $\kappa^{-1}(p)$ with $\kappa(\alpha)$ being defined in \eqref{eq:kappa_for_expectile_def}, which simplifies for location-scale distributions, as follows:
\begin{align*}
\kappa ( \alpha ) 
 &= \frac{\alpha q_Y(\alpha) - \int_{-\infty}^{q_Y(\alpha)} y dF_Y(y)}{ - 2 \int_{-\infty}^{q_Y(\alpha)} y dF_Y(y) - (1-2\alpha) q_Y(\alpha)}.
\end{align*}
This gives us, in the case of the Gaussian distribution (recall the first truncated moment, e.g. from (159) in \cite{Brautigam18_iid_WP}),
\[ \kappa_{norm}(p)  = \frac{p \Phi^{-1}(p) + \phi( \Phi^{-1}(p))}{2 \phi(\Phi^{-1}(p)) - (1-2p) \Phi^{-1}(p)}. \]
For the Student distribution (assumed to be with mean 0, and recalling the first truncated moment computed in (161) in \cite{Brautigam18_iid_WP}),
we obtain
\[ \kappa_{stud} (p) = \frac{ p q_{\tilde{Y}}(p) + \frac{\nu}{\nu-1} f_{\tilde{Y}}(q_{\tilde{Y}}(p)) (1 + q_{\tilde{Y}}^2(p)/\nu)}  {2\frac{\nu}{\nu-1} f_{\tilde{Y}}(q_{\tilde{Y}}(p)) (1 + q_{\tilde{Y}}^2(p)/\nu) - (1-2p) q_{\tilde{Y}}(p) }.\]

For the ES estimator $\widehat{\ES}_{n,t} (p)$ note that it is asymptotically equivalent to $\frac{1}{1-\alpha} \int_{\alpha}^1 q_n(u) du$:
\begin{align*}
\frac{1}{1-\alpha} \int_{\alpha}^1 q_n(u) du &= \frac{1}{1-\alpha} \int_{\alpha}^1 X_{(\lceil nu \rceil)} du = \frac{1}{1-\alpha} \int_0^1 X_{(\lceil nu \rceil)} \1_{(u \geq \alpha)} du
\\ &= \lim_{\Delta x \rightarrow 0}, \frac{1}{1-\alpha} \sum_{i=1}^n X_{( \lceil n x_i^{\star} \rceil)} \1_{(x_i^{\star} \geq \alpha)} \Delta x 
\\ &= \lim_{1/n \rightarrow 0}, \frac{1}{1-\alpha} \sum_{i=1}^n X_{( \lceil n i/n \rceil)} \1_{(i/n \geq \alpha)} 1/n  = \lim_{n \rightarrow \infty}, \frac{1}{n-n\alpha} \sum_{i=1}^n X_{(i)} \1_{(i \geq n\alpha)} 
\\ &= \lim_{n \rightarrow \infty}, \frac{1}{n-n\alpha} \sum_{i=1}^n X_{(i)} \1_{(i \geq \lceil n\alpha \rceil)} 
\\ &= \lim_{n \rightarrow \infty}, \frac{1}{n-n\alpha} \sum_{i=1}^n X_{(i)} \1_{(X_i \geq X_{(\lceil n\alpha \rceil)}}
\end{align*}
where we used this notation to make the connection with the (in this case, by our choice) right Riemann-sum evident. We repartitioned the interval $[0,1]$ into $n$ intervals of length $1/n$, and chose $x_i^{\star}= i/n, i=1,...,n$ to always be the right end-point of each interval.

Note that the asymptotics for the ES estimator $\frac{1}{1-\alpha} \int_{\alpha}^1 q_n(u) du$ where computed in \cite{Brautigam18_iid_WP} (see Section 4.1.3; therein abbreviated as $\widetilde{\ES}_n$).
Because of the more compact integral representation of the asymptotic correlation we keep in the tables the correlation with $\widehat{\ES}_{n,t}(p)$ as in \cite{Brautigam18_iid_WP} (which, as $\widehat{\ES}_{n,t}(p)$ and $\frac{1}{1-\alpha} \int_{\alpha}^1 q_n(u) du$ are equivalent, is equivalent to the representation of the asymptotic correlation in Proposition~\ref{prop:ES-abs-central-moment}).
The explicit solutions of this integral representation are very lengthy and can be found in the Appendix~C of \cite{Brautigam18_iid_WP}.

Let us now present the two tables. First, in Table~\ref{tbl-cor-riskmeasures_norm} the asymptotic correlations for a Gaussian distribution.
\begin{table}[H]
\footnotesize
\begin{center}
\parbox{480pt}{\caption{\label{tbl-cor-riskmeasures_norm}\sf\small Asymptotic correlations between the log-ratios of each, three risk measure estimators, and the two measures of dispersion estimator in the case of a Gaussian distribution}}\\[-1ex]
\hspace*{-2.4cm}
\begin{tabular*}{585pt}{p{3.7cm}|p{6.5cm} p{8.5cm}}
Correlation ~$\displaystyle \lim_{n \rightarrow \infty} \Cor( \hat{m}(X,n,r,t),... )$ &  \PBS\centering Sample Variance &  \PBS\centering Sample MAD
\\[0.5ex]  \hline
\\[-1.5ex] 
 ~...with $\log{\left\lvert \frac{\VaR_{n,t+1y} (p)}{\VaR_{n,t}(p)}\right\rvert}$ & \parbox{6.5cm}{\begin{equation}\label{eq:cor_VaR-var-norm}\frac{-1}{\sqrt{2}}\frac{\phi(\Phi^{-1}(p)) \left\lvert \Phi^{-1}(p)\right\rvert}{\sqrt{2 p(1-p)}}
 \end{equation}} & \parbox{7.5cm}{\begin{equation}\label{eq:cor_VaR-MAD-norm} \frac{-1}{\sqrt{2}}\frac{\left\lvert\phi(\Phi^{-1}(p)) -(1-p) \sqrt{2/ \pi}\right\rvert}{\sqrt{ p(1-p)}\sqrt{1-2/\pi}}
\end{equation}}
\\ & & 
\\ ~...with $\log{\left\lvert \frac{\widehat{\ES}_{n,t+1y} (p)}{\widehat{\ES}_{n,t} (p)}\right\rvert}$ &  \parbox{6.5cm}{\begin{equation}\label{eq:cor_ES-var-norm}\frac{-1}{\sqrt{2}} \frac{ \left\lvert \int_p^1 \Phi^{-1} (u) du \right\rvert}{2\sqrt{\int_p^1 \int_v^1 \frac{v (1-u)}{\phi(\Phi^{-1}(u)) \phi(\Phi^{-1}(v))} du dv}}
\end{equation}} &  \parbox{7.5cm}{\begin{equation}\label{eq:cor_ES-MAD-norm} \frac{-1}{\sqrt{2}} \frac{\left\lvert 1-p - \int_p^1 \frac{1-u}{\phi(\Phi^{-1}(u))\sqrt{2/\pi}} du \right\rvert}
{ 2\sqrt{\left( \frac12-\frac1\pi\right) \int_p^1 \int_v^1 \frac{v (1-u)}{\phi(\Phi^{-1}(v)) \phi(\Phi^{-1}(u))} du dv}}
\end{equation}}
\\ & & 
\\ ~...with $\log{\left\lvert \frac{e_{n,t+1y}(p)}{e_{n,t}(p)}\right\rvert}$ & \parbox{6.5cm}{\begin{equation}\label{eq:cor_expectile-var-norm} \frac{-1}{\sqrt{2}} \frac{\phi(\Phi^{-1}(\kappa^{-1}(p))) \left\lvert \Phi^{-1}(\kappa^{-1}(p))\right\rvert}{\sqrt{2 \kappa^{-1}(p)(1-\kappa^{-1}(p))}}
\end{equation}} & \parbox{7.5cm}{\begin{equation}\label{eq:cor_expectile-MAD-norm} \frac{-1}{\sqrt{2}}\frac{\left\lvert \phi(\Phi^{-1}(\kappa^{-1}(p))) -(1-\kappa^{-1}(p)) \sqrt{2/ \pi} \right\rvert}{\sqrt{ \kappa^{-1}(p)(1-\kappa^{-1}(p))}\sqrt{1-2/\pi}}
\end{equation}} 
\\ \hline
\end{tabular*}
\end{center}
\end{table}
%

The asymptotic correlations, now for an underlying Student distribution with $\nu$ degrees of freedom, are summarised in Table~\ref{tbl-cor-riskmeasures_stud}. 
The expressions look more complex than in the case with the Gaussian distribution. Still, we recover the Gaussian expressions for $\nu \rightarrow \infty$.
For this, recall that $\Gamma(.)$ is the Gamma function, i.e.
\[ \Gamma(x) := \begin{cases}
(x-1)! & \text{for integers } x > 0, \\
\sqrt{\pi} \frac{(2x-2)!!}{2^{\frac{2x-1}{2}}} & \text{for } \text{half-integers $x$, i.e. odd integer-multiples of } \frac{1}{2}, \end{cases} \]
where $!$ and $!!$ denote the factorial and double-factorial function, respectively.
Further, one might need to recall the asymptotic property of the Gamma function$\displaystyle \lim_{n \rightarrow \infty} \frac{\Gamma(n+ \alpha)}{\Gamma(n)n^{\alpha}} = 1$
that we need to use here with $n = \nu \alpha$ and $\alpha =1/2$.

\begin{table}[H]
\footnotesize
\begin{center}
\parbox{480pt}{\caption{\label{tbl-cor-riskmeasures_stud}\sf\small Asymptotic correlations between the log-ratios of each, three risk measure estimators, and the two measures of dispersion estimator in the case of a Student distribution with $\nu$ degrees of freedom}}\\[-1ex]
\hspace*{-2.4cm}
\begin{tabular*}{585pt}{p{3.7cm}|p{7.6cm} p{8cm}}
Correlation ~$\displaystyle \lim_{n \rightarrow \infty} \Cor( \hat{m}(X,n,r,t),... )$ &  \PBS\centering Sample Variance &  \PBS\centering Sample MAD
\\[0.5ex]  \hline
\\[-1.5ex] 
\\  ~...with $\log{\left\lvert \frac{\VaR_{n,t+1y} (p)}{\VaR_{n,t}(p)}\right\rvert}$ & \parbox{7cm}{\begin{equation}\label{eq:cor_VaR-var-stud} \frac{-1}{\sqrt{2}}\frac{f_{\tilde{Y}}(q_{\tilde{Y}}(p)) \left\lvert q_{\tilde{Y}}(p) \right\rvert \left(1 + \frac{q_{\tilde{Y}}^2(p)}{\nu}\right)} {\sqrt{\frac{\nu-1}{\nu-4} \,2\,p(1-p)}} 
 \end{equation}} & \parbox{7cm}{\begin{equation}\label{eq:cor_VaR-MAD-stud}\frac{-1}{\sqrt{2}} \frac{ \left\lvert \frac{\sqrt{\nu (\nu-2)}}{\nu-1} f_{\tilde{Y}}(q_{\tilde{Y}}(p)) \left(1 + \frac{q_{\tilde{Y}}^2(p)}{\nu}\right) - (1-p) \sqrt{\frac{\nu-2}{\pi}} \frac{\Gamma(\frac{\nu-1}{2})}{\Gamma(\nu /2)}  \right\rvert}{\sqrt{p(1-p)} \sqrt{1 - \frac{\nu-2}{\pi}\frac{\Gamma^2(\frac{\nu-1}{2})}{\Gamma^2(\nu /2)}}}
\end{equation}}
\\ & & 
\\  ~... with $\log{\left\lvert \frac{\widehat{ES}_{n,t+1y} (p)}{\widehat{ES}_{n,t}(p)}\right\rvert}$ & \parbox{7cm}{\begin{equation}\label{eq:cor_ES-var-stud}\frac{-1}{\sqrt{2}} \frac{\left\lvert \int_p^1  q_{\tilde{Y}}(u) \left(1+ \frac{q_{\tilde{Y}}^2(u)}{\nu}\right) du \right\rvert}
{2\sqrt{\frac{\nu-1}{\nu-4} \int_p^1 \int_v^1 \frac{v (1-u)}{f_{\tilde{Y}}(q_{\tilde{Y}}(v)) f_{\tilde{Y}}(q_{\tilde{Y}}(u))} du dv} }
\end{equation}} & \parbox{7cm}{\begin{equation}\label{eq:cor_ES-MAD-stud} \frac{-1}{\sqrt{2}} \frac{ \left\lvert \int_p^1 \sqrt{\nu-2} \left(\frac{\sqrt{\nu}}{\nu-1} \, \left(1 + \frac{q_{\tilde{Y}}^2(u)}{\nu}\right) - \frac{\Gamma(\frac{\nu-1}{2})}{\Gamma(\frac{\nu}{2})} \, \frac{(1-u)}{\sqrt{\pi}\,f_{\tilde{Y}}(q_{\tilde{Y}}(u))}\right) \right\rvert}{\sqrt{2 \int_p^1 \int_v^1 \frac{v (1-u)}{f_{\tilde{Y}}(q_{\tilde{Y}}(v)) f_{\tilde{Y}}(q_{\tilde{Y}}(u))} du dv}\, \sqrt{1-\frac{\nu-2}{\pi} \frac{\Gamma((\nu-1)/2)^2}{\Gamma(\nu/2)^2}}}
\end{equation}}
\\ & & 
\\ ~...with $\log{\left\lvert \frac{e_{n,t+1y} (p)}{e_{n,t}(p)}\right\rvert}$ & \parbox{7cm}{\begin{equation}\label{eq:cor_expectile-var-stud}\frac{-1}{\sqrt{2}} \frac{f_{\tilde{Y}}(q_{\tilde{Y}}(\kappa^{-1}(p))) \, \left\lvert q_{\tilde{Y}}(\kappa^{-1}(p)) \right\rvert \left(1 + \frac{q_{\tilde{Y}}^2(\kappa^{-1}(p))}{\nu}\right)} {\sqrt{\frac{\nu-1}{\nu-4} \,2\,\kappa^{-1}(p)(1-\kappa^{-1}(p))}} 
\end{equation}} & \hspace*{-1cm}\parbox{7cm}{ \begin{equation}\label{eq:cor_expectile-MAD-stud} \scriptstyle \quad \frac{-1}{\sqrt{2}} \frac{ \left\lvert\frac{\sqrt{\nu (\nu-2)}}{\nu-1} f_{\tilde{Y}}(q_{\tilde{Y}}(\kappa^{-1}(p))) \left(1 + \frac{q_{\tilde{Y}}^2(\kappa^{-1}(p))}{\nu}\right) - (1-\kappa^{-1}(p)) \sqrt{\frac{\nu-2}{\pi}} \frac{\Gamma(\frac{\nu-1}{2})}{\Gamma(\nu /2)} \right\rvert}{\sqrt{\kappa^{-1}(p)(1-\kappa^{-1}(p))} \sqrt{1 - \frac{\nu-2}{\pi}\frac{\Gamma^2(\frac{\nu-1}{2})}{\Gamma^2(\nu /2)}}}
\end{equation}}
\\ \hline
\end{tabular*}
\end{center}
\end{table}
%

\end{appendices}

\end{document}